\documentclass[11pt, letterpaper]{article}
\usepackage[margin=1in]{geometry}
\usepackage{datetime}

\usepackage{setspace}
\usepackage{float}
\usepackage{amsfonts}
\usepackage{amsmath,amssymb}
\usepackage{amstext}
\usepackage{amsthm}
\usepackage{graphicx}
\usepackage{url}
\usepackage{graphics}
\usepackage{colordvi}
\usepackage{subfigure}
\usepackage{tabularx}
\usepackage{latexsym}
\usepackage{epic}
\usepackage{epsfig}
\usepackage{hyperref}
\usepackage{multirow}% http://ctan.org/pkg/multirow
\usepackage{hhline}
\usepackage{verbatim}
\usepackage[justification=centering]{caption}
\usepackage{diagbox}
\usepackage{enumitem}
\usepackage{color}
\allowdisplaybreaks[1]
\usepackage[numbers]{natbib}
\usepackage{comment}
\usepackage{algorithmic}
\usepackage{lipsum}
\usepackage{booktabs} % For formal tables

\usepackage[linesnumbered,ruled,vlined]{algorithm2e} % For algorithms

\newcommand{\ul}{\underline}

\newtheorem{theorem}{Theorem}[section]
\newtheorem{lemma}[theorem]{Lemma}

\newtheorem{corollary}[theorem]{Corollary}

\newtheorem{proposition}[theorem]{Proposition}

\newtheorem{example}{Example}[section]

\newcommand{\prob}[2][]{\text{\bf Pr}_{#1}\left (#2\right)}

\newcommand{\jj}[1]{{\color{blue} [Junjie: #1]}}
\newcommand{\jjr}[1]{{\color{blue}   #1}}

\newdateformat{monthyeardate}{%
  \monthname[\THEMONTH], \THEYEAR}
  
\begin{document}

\title{The Design of Monopoly Information Broker\footnote{The authors are listed in an alphabetical order. The work was done when Junjie visited Osaka University.}}
\author{Junjie Chen \\ 
City University of Hong Kong \\
\small junjchen9-c@my.cityu.edu.hk
\and Takuro Yamashita \\
Osaka University\\
\small yamashita.takuro.osipp@osaka-u.ac.jp
}
\date{This Version: March 2025}

\begin{titlepage}

\maketitle

\begin{abstract}
An information broker incentivizes consumers to share their information, while designing an information structure to shape the market segmentation. The information broker is a metaphor for an Internet platform that matches consumers with retailers. We are interested in a market with heterogeneous retailers and heterogeneous consumers. The optimal broking mechanism consists of a simple threshold-based structure where consumers with strong preferences are assigned to the efficient retailer while consumers with weaker preferences are assigned to the inefficient retailer stochastically. Our analysis suggests that the privacy protection policy may have a stronger impact on less competitive retail markets.
\end{abstract}

\end{titlepage}

\section{Introduction} \label{sec:introduction}

In a modern economy, large technology firms play a key role in organizing a platform economy, by matching retailers and consumers. Such a platform firm, which we call a \textit{data broker} throughout the paper, often has rich data about consumer characteristics, and can fine-tune the marketing and price recommendations to retailers. It is therefore important to understand the data broker’s incentive and to obtain insights about possible consequences and their welfare implications.

This paper provides a rich yet tractable model of the platform design, with heterogeneous (duopolistic) retailers and heterogeneous (continuously many) consumers. More precisely, the consumers are located at some part of a Hotelling line, and the two retailers are located at the ends of the line, where the location may be interpreted as their preference proximity to each of the retailers.\footnote{For example, consider retailers of different types of computers. Retailer 1 sells high-power machines, while retailer 2 sells less powerful but more energy-efficient machines. The consumers closer to retailer 1 are those who prefer higher power. As they get closer to retailer 2, they put more emphasis on the energy-efficiency.} As explained in the course, allowing for heterogeneity is important, as the mechanism works quite differently for different types of agents. The tractability allows us to characterize the optimal platform mechanisms in various scenarios and to obtain policy implications based on the comparison of those mechanisms.

In particular, we are interested in obtaining economic insights as to \textit{privacy} policies in such a platform market. In practice, there has been an active policy debate whether / how the platform firms should treat the personal data of numerous consumers. On one hand, it is often argued that the platform firms possesses ``too much’’ data about the consumers’ personal information. Some consumers would be unhappy or even economically exploited if their data about online activities are used for unintended commercial purposes. On the other hand, those tech firms advocate that their rich data is useful in enabling customization of retail products and prices for a better treatment of consumers, which would not be possible under strict privacy regulation. It is important to examine this argument from the point of view of a theoretical model, because it is usually difficult to practically attempt various different levels of consumer privacy protection in order to obtain which kinds of private information should be protected (if any).

To examine this argument, we compare two scenarios, where, in one, the data broker has an unlimited access to the consumers data (corresponding to the case of no privacy protection); and in the other, the data broker has no access to it (corresponding to the case of full privacy protection). Theoretically, the first case corresponds to a benchmark model where the data broker knows the consumers' types and hence his mechanism is free from the consumer's truth-telling constraint. The second case corresponds to our main problem where the truth-telling constraint (as well as the retailers' obedience constraints) must be satisfied, because, without any access to the consumers' data, the only way for the mechanism to be contingent on it is to make the consumers' voluntary disclosure incentive compatible.

The results (Theorem \ref{thm: ic+ob} and Proposition \ref{prop: ic}) show that the effect of privacy protection is heterogeneous among the consumers. For those who have relatively stronger preference toward one retailer than the other, the privacy protection would be beneficial: they are the types of consumers who would much enjoy buying from the preferred retailer, while their high valuations would be extracted without privacy protection. With privacy protection, they can keep some of their values as their ``information rent''. %\jjr{[it means in the case where the broker knows all the information, the broker can extract all the consumer's surplus?]}. 
In contrast, for those who are closer to be indifferent across the retailers, they cannot expect much increase of the value under privacy protection. %\jjr{[you means it is because those consumers close to $\frac{1}{2}$, their information rent is $0$?]} 
Rather, their allocation would be inefficiently biased (by being navigated toward buying from inefficient/less-favored retailers) in order to reduce the information rent to the stronger-preference types.

We also observe that competition among the retailers plays an important role in determining the eventual effect of privacy protection. To see this, we examine the case of a monopolistic retailer for comparison (with the baseline duopolistic case).\footnote{As it turns out, this case would be equivalent to the case where the data broker can dictate the price choice of the retailers.} 
Interestingly, while qualitatively a similar pattern being observed, the effect of privacy protection would be more nuanced in case of duopolistic retailers than in the monopoly case: (i) Recall that, for the closer-to-indifferent consumers, the allocation would be inefficient because they are stochastically navigated toward buying from inefficient retailers. In case of duopolistic retailers, the efficient retailer's incentive (formally, his \textit{obedience constraint}) puts a significant upper bound on how likely this inefficiency happens. That obedience constraint is null in case of the monopoly retailer, %\jjr{[but in the monopoly case, there is no obedience constraints, right?]},
and hence more inefficiency would happen. (ii) For the stronger-preference consumers, recall that the inefficiency for the closer-to-indifferent consumers is in order to reduce the information rent given to the stronger-preference consumers. Thus, more inefficiency in case of the monopoly retailer implies less information rent for the stronger-preference consumers accordingly, than in case of the duopoly retailers.

To conclude, the privacy protection policy would have a stronger impact in case of a less competitive retail market. %The regulatory authority should keep a close look at such less competitive markets;\jjr{(why?) more competitive markets would enjoy more values created by the fine-tuned matching technology.} 
Overall, our analysis suggests importance of assessing the retail market structure, in order to appropriately evaluate the pros and cons of the platform issues (and their welfare consequences).

\subsection{Related Works} \label{sec:literature}

Our paper contributes to a recent literature on how data and information influence the competitive market. Similar to ours, \citet{ichihashi2021competing} studies the data brokerage problem but with a specific focus on the competition among multiple data brokers. They consider data as nonrivalry goods so that it can be shared with multiple brokers. On the contrary, our paper takes the information design approach to analyze how the monopolistic broker optimally reveals consumers' information to competing retailers in the downstream market. 

Much of the existing literature on this line assumes that the data seller can fully access the consumers' information (i.e., our ``no privacy protection'' case): see \citet{elliott2021market}, \citet{BONATTI2024105779}, and \citet{bounie2021selling}, for example. \citet{elliott2021market}  study how an information designer can reshape the segmentation in a competitive market through information, where the designer can exclusively observe the state of nature without any effort. They show that in competitive setups, the designer can achieve any efficient solution between the consumer-optimal information structure and producer-optimal information structure. \citet{BONATTI2024105779} study the information selling problem to information buyers with private types. \citet{bounie2021selling} employ the Hotelling model like in our paper and study an information selling problem. Contrary to those papers, our model considers a problem where the data broker needs to incentivize the consumers to share their information before revealing it to the downstream market. Comparing the cases with and without the consumers' incentive compatibility, we obtain some insights as to the privacy protection policy.

This paper is also related to the literature on the sale of information and data. 
\citet{bergemann2015selling} study the data provider's pricing problem where data is used for targeted marketing. \citet{mehta2019sell} study a more concrete problem where the data seller sells data records stored in a database to a buyer. \citet{chen2022selling} and \citet{agarwal2019marketplace} study the data selling problem in the specific machine learning scenario. Another line of data selling method takes the information design approach. For example, \citet{bergemann2018design} explores the pricing of information where the seller sells experiments to a decision-maker who has private information. \citet{yangselling} considers a seller selling consumers' information to a monopoly retailer, where the seller is assumed to be powerful enough to create any market segmentation. \citet{smolin2023disclosure} considers a monopolist selling one object to a buyer by pricing the attributes so as to influence the buyer's willingness. 
On the contrary, \citet{ichihashi2022collection} explore an opposite problem where the seller collects information (in an information design way) from the consumers to better price the goods where the information requests are modeled as signals.

Finally, this paper is related to the Bayesian persuasion. \citet{kamenica2011bayesian} first study the Bayesian persuasion problem where one sender reveals signals to influence the decisions of one Bayesian receiver. Later, \citet{arieli2019private} extend to the setups of multiple receivers who communicate with the sender privately. \citet{gentzkow2016rothschild} characterize the optimal solution of Bayesian persuasion when only the posterior mean matters. \citet{dworczak2019simple} investigate a similar problem by establishing the connections to First Welfare Theorem. \citet{guo2019interval} prove the interval structure of optimal mechanism in an environment where the agent has binary actions. Recently, \citet{smolin2022information} study an information design problem in a game with players' payoffs being concave.

\section{Preliminaries} \label{sec:model}

%\jj{Use $V_1$, $V_2$ instead of $V$. To be done.}

Two sellers $S_1$ and $S_2$ are located at the ends of a Hotelling line $[0,1]$: $S_1$ is located at $0$ and $S_2$ is located at $1$. Each seller can set either a high price $p =H$ or a low price $p = L$ for its product, with $L<H$. This binary-price assumption is restrictive but significantly simplifies the problem.\footnote{The case of a continuous price space would be more ideal and is left for future research, though much less tractable.} Let $\Omega=\{(L,L),(L,H),(H,L),(H,H)\}$ denote the space of possible prices set by the sellers.

There is a unit mass of consumers distributed over the Hotelling line according to some distribution $\Lambda: [0, 1] \to [0, 1]$, whose density (or probability mass) function $\lambda$. A consumer at location $x \in [0, 1]$ (or more simply, a ``consumer $x$'') can choose to purchase a product either from $S_1$, who locates at $0$, or from $S_2$, located at $1$. The consumer's location can be interpreted either as the physical distance to each seller's, or as representing the consumer's \textit{bliss point} of the product characteristics. Specifically, the consumer $x$'s payoff is $V_1-xt-p_1$ if he buys from $S_1$ with price $p_1$, while it is $V_2-(1-x)t-p_2$ if he buys from $S_2$ with price $p_2$, for some parameter $t>0$. For example, if $x=0$, then he would enjoy the gross payoff $V_1$ from $S_1$'s product, either because he does not need any transportation cost to travel to $S_1$ or his bliss point is exactly where $S_1$'s product is at. On the other hand, his gross payoff from $S_2$'s product is only $V_2-t$, as his location is far from $S_2$'s. As $x$ becomes higher, his gross payoff from $S_1$ becomes smaller, and that from $S_2$ becomes larger.

Throughout the main part of the paper, we assume that $V_1=V$ and $V_2=V-t$ with $V-t-H>0$ so that all consumers would prefer $S_1$ to $S_2$ \textit{if the prices of both sellers are the same} (and prefer buying some product to not buying at all, even if with the high price).\footnote{The case with $V_1-V_2\geq t$ would be qualitatively similar.} In this sense, $S_1$ is in a more dominant position than $S_2$. The symmetric case would be equally interesting and briefly discussed at the end of the paper, but it seems a lot more complicated. Also, some markets in practice indeed have some dominant sellers, and in this sense, this asymmetric case has its own right.

%Especially with this second interpretation, it is reasonable to assume that this precise location $x$ is the consumer's private information without any data about his past online activities.  

The sellers and consumers can match and trade on an online platform, organized by a (monopoly) \textit{data broker} (e.g., Amazon in US and Alibaba in China). The role of the data broker is to communicate with each side of the market and propose an allocation. Let $\pi:[0,1]\to \Delta(\Omega)$ represent the data broker's \textit{(direct) communication strategy}, with the interpretation that, for each $(s_1,s_2)\in \Omega$ and $x$, $\pi((s_1,s_2)|x)$ denotes the probability of the data broker's privately recommending price $s_i$ to each seller $i=1,2$ (to be more precise, $s_i$ is only observed by seller $i$) in case the consumer is located at $x$. The data broker can also charge a \textit{participation fee} (paid to the platform) $m_c(x)$ to the consumer and $m_i$ to each seller $i$. Without the fee, it is not possible for a player to participate in the platform, yielding a payoff of 0 for normalization. The tuple $(\pi,m_c,m_1,m_2)$ is called a \textit{(direct) mechanism}.

Regarding the data broker's knowledge about $x$, we consider two scenarios (and later compare them): In one scenario, the data broker knows each consumer's location perfectly, based on the idea that those consumers' preferences can precisely be estimated using the large data about their online activities. We also consider an alternative scenario, where those consumers' data are protected as private (``privacy protection''), where the appropriate \textit{incentive compatibility} conditions must be satisfied in order to elicit the consumers' preferences through their voluntary disclosure. In both scenarios, each seller is free to choose whichever price he likes, $L$ or $H$. More specifically, consider the following timing of actions.

\begin{enumerate}
\item The data broker commits to a mechanism $(\pi,m_c,m_1,m_2)$. The consumer and retailers must pay their participation fees to the data broker: the consumer pays $m_c(\tilde{x})$ and each seller $i$ pays $m_i$ (refusing participation implies the payoff of 0).
\item Consumer $x$ reports his location $\tilde{x}$ to the data broker. Note: In the scenario where the data broker knows the consumer location, necessarily $x=\tilde{x}$; while in the other scenario, $\tilde{x}$ can be different from $x$.
\item $s=(s_1,s_2)$ is drawn according to $\pi(s|\tilde{x})$, and $s_i$ is recommended to seller $i$.
\item Each seller $i$ chooses its retail price $\tilde{s}_i$. Note: $\tilde{s}_i$ can be different from $s_i$.
\item The consumer chooses which seller to buy from. A trade takes place.
\end{enumerate}

\textbf{Sellers' Payoff.} After observing signal $s_i$ drawn according to $\pi(s|x)$, seller $S_i$ will update his belief about consumer's location $x$ and signal $s_{-i}$ of seller $S_{-i}$, i.e., 
\[\prob{x, s_{-i}|s_i} = \frac{\pi(s_i, s_{-i}|x)\lambda(x)}{\sum_{s_{-i}}\int_{x} \pi(s_i, s_{-i}|x)\lambda(x)dx} \]
Based on our previous argument, it is without loss of generality to consider a direct scheme where it is optimal for $S_i$ to always follow the recommended price $s_i$. This optimality condition can be written as the following \textit{obedience constraint}:
\begin{equation}\label{obedience_con}
\begin{aligned}
    \sum_{s_2} s_1 \int_{0}^1 {\bf 1} &\Big\{ V_1-xt -s_1\ge V_2-(1-x)t -s_2 \Big\}\pi(s_1, s_{2}|x)\lambda(x) dx \\
    & \ge \sum_{s_2} {s}'_1 \int_{0}^1 {\bf 1} \Big\{ V_1-xt -{s}'_1\ge V_2-(1-x)t -s_2 \Big\}\pi(s_1, s_{2}|x)\lambda(x) dx, \quad \forall s_1, s'_1 \in \{H, L\}
\end{aligned}
\end{equation}
where the left-hand side of (\ref{obedience_con}) is seller $S_1$ utility by following the price $s_1$ while the right-hand side is the utility by deviation. Hence, the expected revenue of seller $S_1$ is 
\[U_1 = \sum_{s_1} \sum_{s_2} s_1 \int_{0}^1 {\bf 1} \Big\{ V_1-xt -s_1\ge V_2-(1-x)t -s_2 \Big\}\pi(s_1, s_{2}|x)\lambda(x) dx \]
Similar obedience constraints and expected revenue  $U_2$ are defined for $S_2$. 
Finally, the broker charges the seller $S_i$ the participation fee $m_i\ge 0$. Hence, the payoff for seller $S_i$ is $U_i - m_i$. The \textit{individual rationality} constraint ensures that the sellers gain nonnegative expected utilities by participating in the mechanism:
\begin{equation}\label{sellerindrational_con}
    U_i - m_i \ge 0.
\end{equation}

\textbf{Consumer Payoff.} Let $U(x, x')$ be the utility of consumer $x$ in case he reports $x'$ as his location:
\[
U(x, x') = \sum_{s_1, s_2} \pi(s_1, s_2|x')\max \{V_1-tx -s_1, V_2-t(1-x)-s_2\}
\]
Define $U(x) \triangleq U(x, x)$. To incentivize the consumers to be truthful, we employ the \textit{incentive compatibility} constraint for the consumer $x$:
\begin{equation}\label{icenticomp_con}
    U(x) -m_c(x) \ge U(x, x') -m_c(x'), \quad \forall x, x' \in [0, 1]
\end{equation}
where $m_c(x)\ge 0$ is the fee charge designed by the broker for the consumer who claim to locate at $x$. 
Also, the consumer's individual rationality constraints are imposed: 
\begin{equation}\label{irconconsumer}
    U(x)-m_c(x)\ge 0.
\end{equation}

The transfers $m_c(x)$ and $m_i$ from consumers and sellers to the broker correspond to the broker's monopoly power. 
For example, in order to sell products through online platforms, sellers usually need to pay the platforms, e.g., selling fees and referral fees on Amazon \cite{Howmuchd35:online}. These fees are described by the transfer $m_i$ in our model.
The users in practice (consumers in our model) may or may not be charged the fee. The most popular business model seems to set the baseline registration for free, and charge some fees for ``premium'' versions. In our model, $\underline{m}_c=\inf_x m_c(x)$ may be interpreted as the baseline fee, and $m_c(x)-\underline{m}_c$ can be interpreted as the fee for the premium for the version consumer $x$ chooses. It is not necessarily that $\underline{m}_c=0$, and in this sense, does not perfectly match the reality. However, the literature points to some behavioral motifs for such a free-baseline business model. We conjecture that the paper's qualitative messages would be robust with a behavioral twist which induces $\underline{m}_c=0$.

\textbf{Broker's revenue-maximizing problem} The broker is to maximize his own revenue by designing a mechanism, which is formulated as
\begin{equation}\label{prog_broker}
\begin{aligned}
    \max_{\pi, m} \quad & m_1+m_2 + \int_{0}^{1} \lambda(x) m_c(x) dx\\
    \textnormal{subject. to} \quad & (\ref{obedience_con}), (\ref{sellerindrational_con}), (\ref{icenticomp_con}), (\ref{irconconsumer})
\end{aligned}
\end{equation}

{\textbf{Market Segmentation.}} By comparing the utilities of purchasing from seller $S_1$ and seller $S_2$, i.e., $V_1-xt-p$ and $V_2-(1-x)t-p$, we can divide the whole market into 2 segments. 

%we can divide the whole market into $4$ segments. See Figure \ref{signal_region_painting} for an explanation.
%\begin{definition}
%    The indifferent location $\bar{x}$ is defined as that the consumer at $\bar{x}$ is indifferent between purchasing products from $S_1$ and $S_2$.
%\end{definition}
%Hence, we have $3$ indifferent locations $\bar{x}$: $\frac{1}{2}$, $\frac{1}{2}-\frac{H-L}{2t}$ and $\frac{1}{2}+\frac{H-L}{2t}$. These indifferent locations divide the market into $4$ segments: $[0, \frac{1}{2}-\frac{H-L}{2t}]$, $[\frac{1}{2}-\frac{H-L}{2t}, \frac{1}{2}]$, $[\frac{1}{2}, \frac{1}{2}+\frac{H-L}{2t}]$ and $[\frac{1}{2}+\frac{H-L}{2t}, 1]$, where the consumers in the same segment have the same purchasing behavior. 

%\begin{figure}
%    \centering
%\includegraphics[width=0.3\textwidth]{regionsignal.pdf}
%    \caption{There are $4$ signals, and the three indifferent locations divide the market into $4$ segments. The consumers in the gray regions will purchase products from seller $S_1$.}
%    \label{signal_region_painting}
%\end{figure}

\section{Privacy Protection and Competition}

In this section, we examine the effects of privacy protection with duopoly retailers.

%\jj{begin. Should we keep this paragraph to set up the model we use? This will help the illustration, I think.}

%\jjr{
%In the main part of the analysis, we assume that retailer $S_1$ is more dominant in this market and consider an {\it asymmetric}  situation where all consumers prefer the products of seller $S_1$ to that of $S_2$, if they are offered at the same prices. In other words, the consumer will enjoy value $V_1$ if buying products from $S_1$ and $V_2$ if if buying products from $S_2$, where $V_1 \gg V_2$.  Indeed, our main Theorem \ref{thm: ic+ob} holds generally for $V_1 - V_2 \ge t$. To ease exposition, we  consider the case where 
%\[V_1 = V ~\text{and} ~V_2=V-t\]
%With this setting, the market segmentation is changed as in Figure \ref{signal_region_painting2}. All the consumers are {\it uniformly} located at $[0, 1]$. We call such a model as {\it Asymmetric Hotelling}.  This is equivalent to the case where $V_1 =V_2 =V$ but all the consumers are only {\it uniformly} located at $[0, \frac{1}{2}]$, which we dub as {\it Symmetric Hotelling}.
%We later discuss the symmetric case briefly, but some markets may better be described as this type of asymmetric market than symmetric ones since the qualities of products usually vary in practice.

\begin{figure}
    \centering
\includegraphics[width=0.3\textwidth]{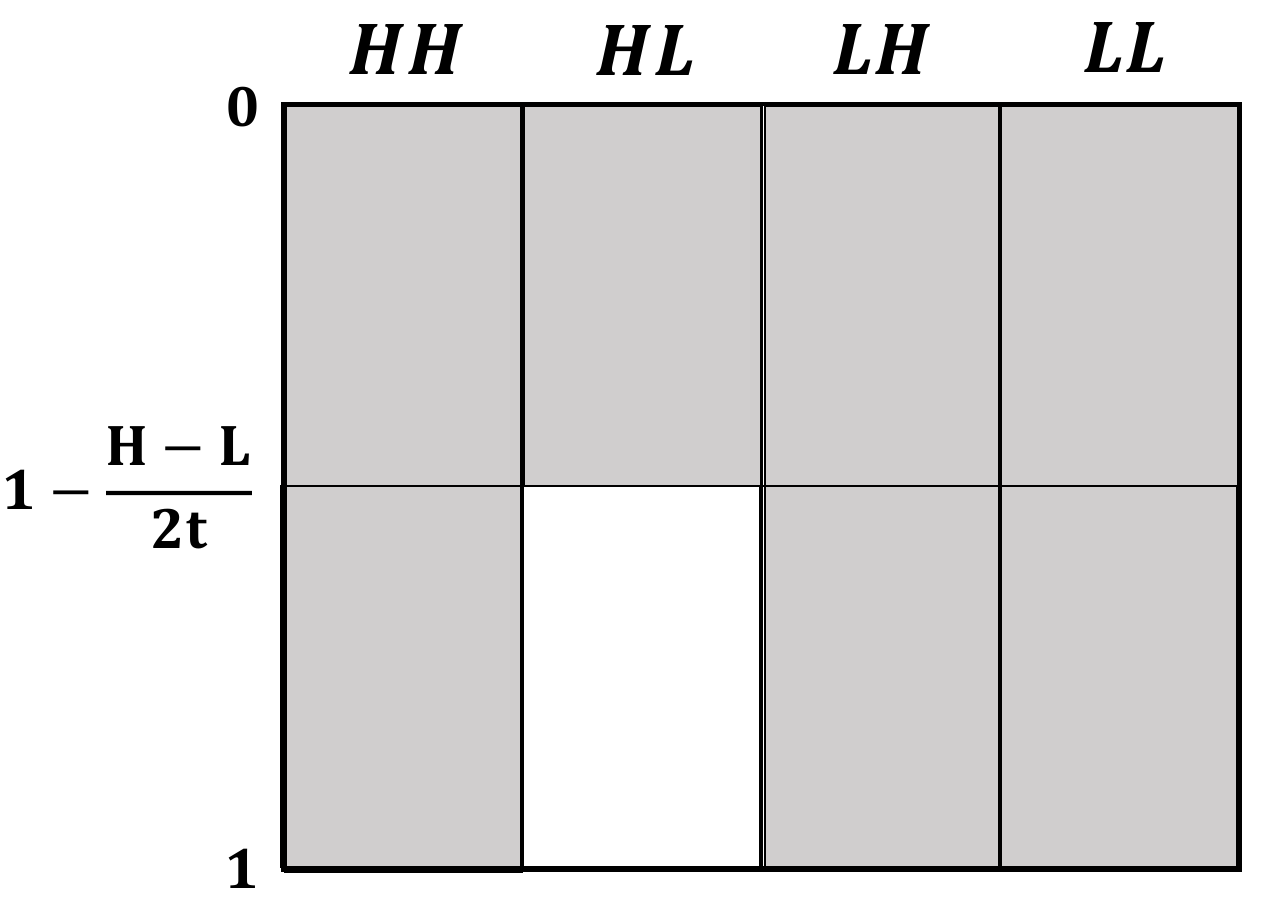}
    \caption{There is one indifferent location dividing the market into $2$ segments.}
    \label{signal_region_painting2}
\end{figure}

%}

%\jj{I believe that main   theorem 3.2 indeed holds for $V_1 - V_2 \ge t$. To verify.}\jj{checked. I believe it is true.}

\subsection{Without Privacy Protection}\label{sec: ob}

First, we analyze the simpler case where there is no privacy protection. In this case, a mechanism does not have to satisfy the consumer's truth-telling incentive constraints, and hence only the retailers' obedience constraints are in place.

The following mechanism is optimal in this case.

\begin{proposition}
For $x<\ul{x}$, the price offers are $(H,H)$, and they buy from $S_1$. For $x>\ul{x}$, the price offers are $(L,L)$, and they buy from $S_1$. The broker charges the fees to retailers and consumers so that their payoffs are all 0.\footnote{More specifically, $m_1=...$ for $S_1$, $m_2=0$ for $S_2$, $m_c(x)=$ for the consumers $x<\ul{x}$, and $m_c(x)=$ for the consumers $x>\ul{x}$.}
\end{proposition}

%All the proofs including the one for this proposition are in the appendix. 
To provide the intuition, notice that $\ul{x}$ represents the threshold type consumer who is indifferent between (i) buying from $S_1$ with price $H$ and (ii) buying from $S_2$ with price $L$. For the consumer $x<\ul{x}$, buying from $S_1$ is a dominant strategy. Thus, it is natural that $S_1$ charges $H$ (and it does not matter what $S_2$ charges). For the consumer $x>\ul{x}$, the two retailers are in a real competition (and without privacy protection, the retailers know it). Therefore, the only possible outcome is that both charge $L$.

Note that the equilibrium outcome is efficient: all the consumers buy from $S_1$, and hence minimizing the wasteful transportation cost. Also, all the welfare is captured by the data broker, and the other players' payoffs are all 0.

\subsection{With Privacy Protection}\label{sec: ic+ob}

Next, we analyze the more complicated case with privacy protection, where a mechanism must satisfy both the incentive compatibility and obedience constraints.

First, observe that the mechanism found in the previous proposition (i.e., the optimal mechanism without privacy protection) is now infeasible, because it violates the consumer's incentive compatibility constraint. In other words, if the consumer's privacy is protected, then the previous mechanism does not successfully elicit the consumer's information. Specifically, consider the consumer with $x<\ul{x}$. In the previous mechanism, both retailers offer $H$. However, imagine that this consumer pretends to be of type $x'>\ul{x}$. Then this consumer is offered the lower price $L$, and hence is better off.

The optimal mechanism with privacy protection is given as follows.
\begin{theorem}\label{thm: ic+ob}
For $x<\ul{x}$, the price offers are $(H,H)$, and they buy from $S_1$. For $x>\ul{x}$, there is a threshold $x^*\in[\ul{x},1]$ such that: (i) for $x\in(\ul{x},x^*)$, the price offers are $(L,L)$, and they buy from $S_1$; (ii) for $x>x^*$, the price offers are either $(L,L)$ or $(H,L)$ equally likely, and they buy from $S_1$ given $(L,L)$, while buy from $S_2$ given $(H,L)$.

The broker's fees to the retailers are such that their payoffs are 0. The fees to the consumers are such that (only) those with $x<\ul{x}$ earns strictly positive payoffs (\textit{information rents}).
\end{theorem}

First, for the consumers with $x<\ul{x}$, they face the same prices $(H,H)$ as before. However, they pay less to the platform, so that they earn some information rent.

In order to save this information rent, the optimal mechanism introduces some inefficiency in the allocation. The basic idea is to make some consumers stochastically buy from $S_2$ instead of $S_1$ through recommending the prices $(H,L)$. Buying from $S_2$ requires more transportation cost and hence more inefficient. 

Indeed, for the consumers with $x>\ul{x}$, there is an additional threshold $x^*$: for those with $x>x^*$, it is equally likely that the prices are $(L,L)$ or $(H,L)$; while for those with $x\in(\ul{x},x^*)$, the retailers offer $(L,L)$.

To explain the allocation for $x>x^*$, notice that the consumers with $x>x^*$ are relatively close to be indifferent between the goods from $S_1$ and $S_2$. Buying from $S_2$ requires slightly more transportation cost but not much. This introduction of inefficiency is effective in preventing the deviation from those closer to $S_1$, because for those types, buying from $S_2$ would require much larger transportation costs. 

In contrast, those with $x\in(\ul{x},x^*)$ are too far from $S_2$ to introduce $(H,L)$ with some probability. Thus, the optimal mechanism continues to assign $(L,L)$ for sure as in the case without the truth-telling constraints.

So far, we explain how the optimal allocation is influenced by the consumers' truth-telling constraints due to privacy protection. The basic logic is the standard rent-efficiency tradeoff. However, there is also an influence by the retailers' obedience constraints, and their interplay is key to understand our results. As we see later, the obedience constraints play its main role in the determination of the threshold $x^*$, but let us defer its explanation until the next section when we discuss a similar problem but without the obedience constraints.

\subsubsection{Idea of the proof}

\begin{figure}[t]
  \centering
  \includegraphics[width=4.5in]{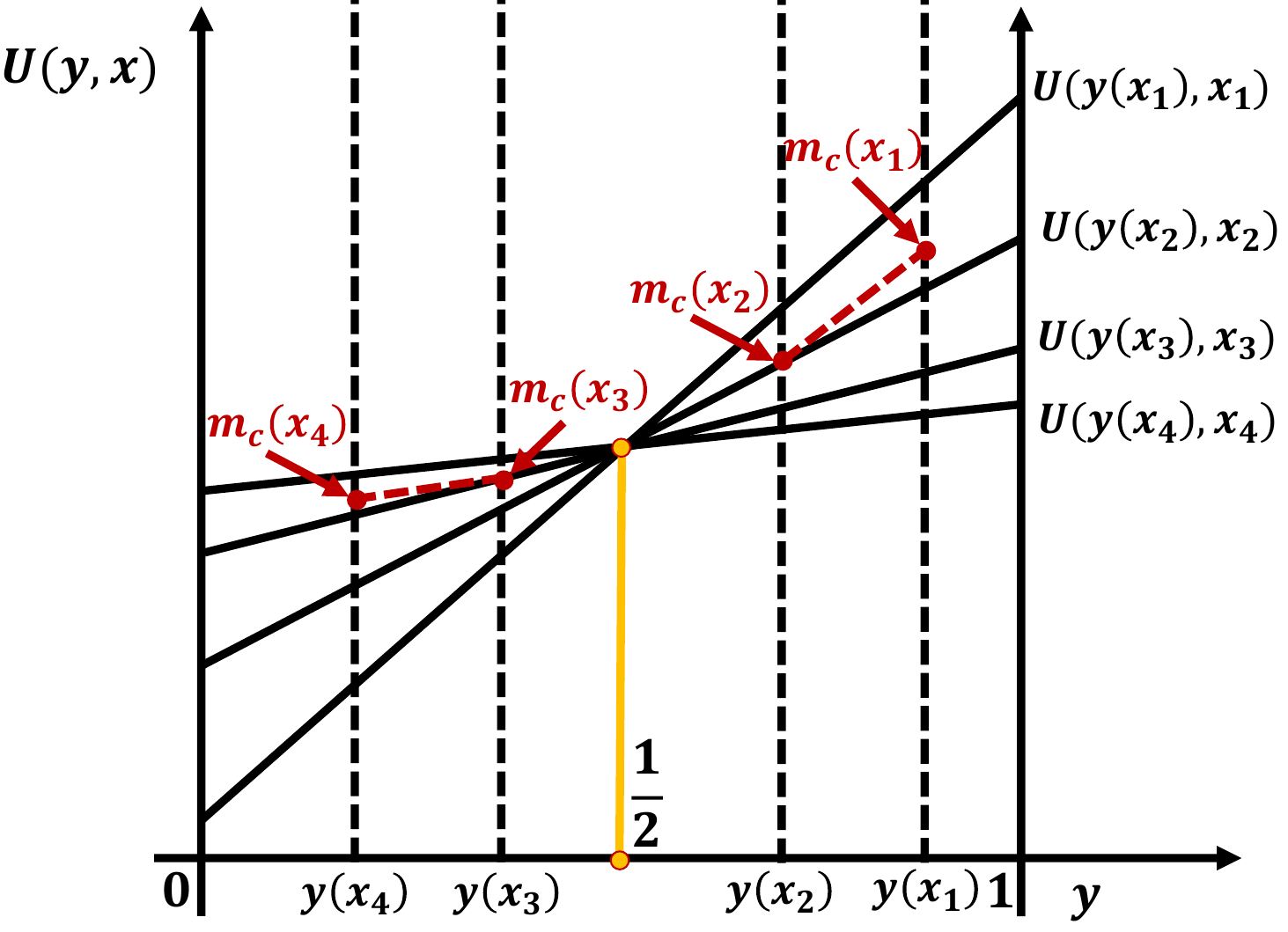}
  \caption{In this example, there are four consumer $x_1 < x_2 <x_3 <x_4$. The utility $U(y(x), x)$ for all $x$ has the same value at $y(x)=\frac{1}{2}$. The IR constraints for $x_2$ and $x_3$ bind at the optimal design. The IC constraint between $x_1$ and $x_2$ and the IC constraint between $x_3$ and $x_4$ bind.} 
  \label{fig:brokergraphd}
\end{figure}

The remaining part of this section is devoted to the intuition for the optimal information structure presented in Theorem \ref{thm: ic+ob}, in order to provide some idea of the proof. A fully formal proof can be found in Section \ref{sec:opt-one-side}. Let $y(x) \triangleq \pi((L,L)|x), \forall x\in [\underline{x}, \overline{x}] $ be the probability of sending signal $(L, L)$. We aim to show two key properties in Theorem \ref{thm: ic+ob}: i) $y(x) \ge \frac{1}{2}$ for all $x\in [\underline{x}, \overline{x}]$, and ii) there exists some $x^* \in [\underline{x}, \overline{x}]$ such that $y(x)=1$ if $\underline{x} \le x \le x^*$ and $y(x)=\frac{1}{2}$ if $ x^* \le x \le  \overline{x} $. 

To illustrate the first property, consider an example where there are four consumers $x_1, x_2, x_3, x_4 \in [\underline{x}, \overline{x}]$ with $x_1 < x_2< x_3< x_4$, and a feasible solution with $y(x_4) < y(x_3) < \frac{1}{2} < y(x_2) < y(s_1)$. The valuation for consumer $x \in [\underline{x}, \overline{x}]$ with allocation $y$ is 
\[
U(y, x) = y[V-xt-L] + (1-y)[(V-t) - (1-x)t -L],
\]
which is a linear function of $y$. Because the consumers' valuation functions intersect at one point where $y=\frac{1}{2}$, to maximize the broker's revenue, both the consumers $x_2, x_3$ have the {\it zero} utility, and the IC constraint between $x_1$ and $x_2$ (also $x_3$ and $x_4$) is binding. 
See Figure \ref{fig:brokergraphd} for illustration. The binding IC constraints imply that the segment $(m_c(x_1), m_c(x_2))$ has the same slope as $U(y(x_1), x_1)$, and the segment $(m_c(x_3), m_c(x_4))$ has the same slope as $U(y(x_4), x_4)$.
Hence, the optimal design always follows the implication of ``no information rent for the lowest type". Note that here, the two lowest types are $x_2$ and $x_3$, as the order of valuation functions reverses at $y=\frac{1}{2}$. We can easily see that the allocations $y(x)$ (with the payments $m_c(x)$) are not optimal. Indeed, one can construct a new solution by increasing to $y(x_3)=y(x_4)=\frac{1}{2}$ and set payments such that {\it zero} utility are left to $x_3, x_4$. This new solution is obviously feasible and will generate larger revenue for the broker.

\begin{figure}[t]
  \centering
  \includegraphics[width=4in]{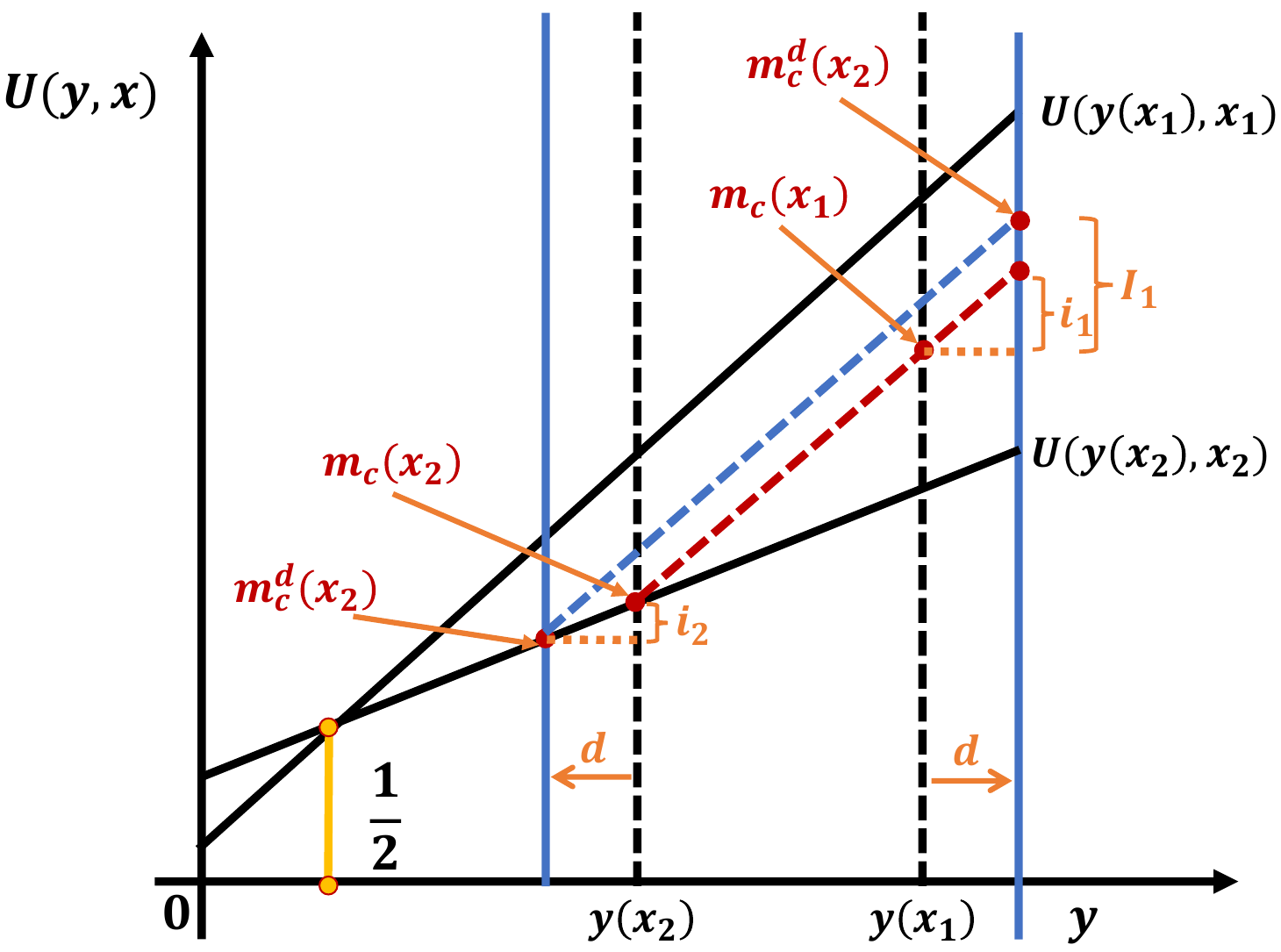}
  \caption{The allocation probabilities are adjusted (increase and decrease) by $d$. It generates a new solution with new maximum payments $m^d_c(x_1), m^d_c(x_2)$. The segment $(m^d_c(x_1), m^d_c(x_2))$ has the same slope as $U(y(x_1), x_1)$ and segment $(m_c(x_1), m_c(x_2))$.} 
  \label{fig:split}
\end{figure}

\begin{figure}[t]
  \centering
  \includegraphics[width=3.5in]{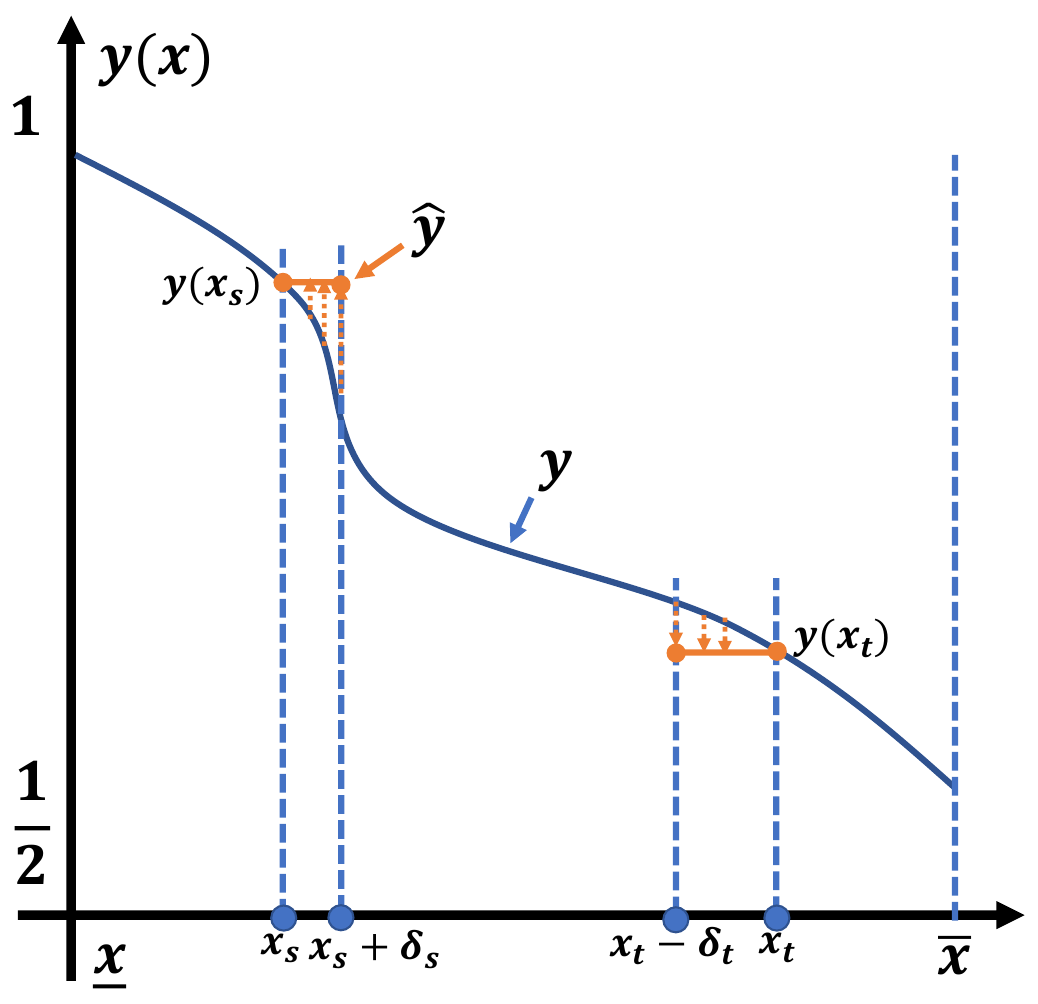}
  \caption{In this example, the probability function $y(x)$ with $x\in [x_s, x_s+\delta_s]$ is moved up to $y(x_s)$, while the probability function $y(x)$ with $x\in [x_t-\delta_t, x_t]$ is moved down to $y(x_t)$.} 
  \label{fig:threshold}
\end{figure}

The second property comes from the structure of the optimal design. We take a closer look at the $x_1, x_2$ consumers in the example of Figure \ref{fig:brokergraphd}. We assume $\lambda(x_1) = \lambda(x_2)$. If decreasing $y(x_2)$ by $d$ and increasing $y(x_1)$ by $d$, by the discussion for the first property, we obtain new maximum payments $m^d_c(x_1), m^d_c(x_2)$. Note that the segment $(m^d_c(x_1), m^d_c(x_2))$ has the same slope as $U(y(x_1), x_1)$. See Figure \ref{fig:split} for illustration. By increasing $y(x_1)$ with $d$, the payment from $x_1$ increases $i_1$, while the payment from $x_1$ decreases $i_2$ by decreasing $y(x_2)$ with $d$. $i_1 > i_2$ since the slope $U(y(x_1), x_1)$ is greater than the slope of $U(y(x_2), x_2)$. Note that decreasing $y(x_2)$ would further increase the payments from $x_1$, i.e., $I_1 > i_1$. This implies that through increasing some probabilities $y(x)$ of some consumers and decreasing the probabilities of some other consumers by the same amount, the broker's revenue increases. 
Similar intuitions work in continuous settings, as shown in Figure \ref{fig:threshold}. Given any allocation $y(x)$ for all $x \in [\underline{x}, \overline{x}]$, we can show that for any two $x_s, x_t$ and $\delta_s, \delta_t >0$ such that $\int_{x_s}^{x_s+\delta_s} [y(x_s) - y(x)] dx = \int_{x_t-\delta_t}^{x_t} [y(x) - y(x_t)] dx$, i.e., the probability masses are the same, by increasing $y(x)$ to $y(x_s)$ for all $x \in [x_s, x_s+\delta_s]$ and decreasing $y(x)$ to $y(x_t)$ for all $x \in [x_t-\delta_t, x_t]$ and decreasing, the broker's revenue increases. By moving all the probabilities towards $\frac{1}{2}$ and $1$ with this operation, the optimal structure will finally be implemented by a single threshold $x^*$.

%In contrast to the standard mechanism design, however, the competitive pressure plays a crucial role at this point of the logic. As we discuss later, the retailers' obedience constraints put a significant \textit{upper bound} on how much inefficiency can be introduced. 

%In order to better explain the interplay between the competitive pressure (through the obedience constraints) and the motif of introducing inefficiency for information rent saving, it is instructive to first consider another (benchmark) scenario where the consumers' truth-telling constraints must be met but not the retailers' obedience constraints. Then we come back to the interpretation of Theorem \ref{thm: ic+ob} through the lens of this benchmark scenario.

\subsection{Effect of Market Structure}\label{sec: ic}

So far, we assume duopoly retailers, and see how the privacy protection (and its resulted information rents) affects the optimal mechanisms.

In this section, we consider the case where the data broker also plays the role of retailers and choose the retail prices freely (i.e., without the retailers' obedience constraints). Indeed, in some online marketplaces, sellers only trade with the online platform, who then sells the goods to the consumer on its marketplace. A possible interpretation is a stronger market power of the platform over the sellers. Formally, this means that the broker designs the mechanism without the obedience constraints (but with the consumers' truth-telling constraints, in the case of privacy protection).

The motivation of this exercise is two-fold. First, this problem can serve as a benchmark for better understanding of Theorem \ref{thm: ic+ob}. Indeed, the comparison below clarifies the precise role the obedience constraints played in the determination of the optimal mechanism in Theorem \ref{thm: ic+ob}.

Second, the comparison uncovers the role of the \textit{competitiveness} of the market. The case where duopoly retailers actively engage in the price selection seems to enjoy more competitive pressure than the case where the broker fully dictates the prices. The analysis is policy relevant, as it would provide some insights as to which markets the regulatory authority should pay more attention to, especially in relation to the introduction of the privacy protection policy.

The optimal mechanism now becomes as follows.
\begin{proposition}\label{prop: ic}
For $x<\ul{x}$, the price offers are $(H,H)$, and they buy from $S_1$. For $x>\ul{x}$, there is a threshold $x^{**}=\max\{\ul{x},\frac{1}{2}\}\leq x^*$ such that: (i) for $x\in(\ul{x},x^{**})$, the price offers are $(L,L)$, and they buy from $S_1$; (ii) for $x>x^{**}$, the price offers are either $(L,L)$ or $(H,L)$ equally likely, and they buy from $S_1$ given $(L,L)$, while buy from $S_2$ given $(H,L)$.

The broker's fees to the retailers are such that their payoffs are 0. The fees to the consumers are such that (only) those with $x<\ul{x}$ earns strictly positive payoffs (\textit{information rents}), while smaller than the mechanism in Theorem \ref{thm: ic+ob}.
\end{proposition}

The shape of the optimal mechanism resembles the one in Theorem \ref{thm: ic+ob}. First, the consumers with $x<\ul{x}$ are charged $(H,H)$, and they enjoy some information rents due to the privacy protection. For those $x>\ul{x}$, there is a threshold type, in the current case $x^{**}$, such that those with $x>x^{**}$ are assigned the prices $(H,L)$ or $(L,L)$ equally likely; while those with $x\in(\ul{x},x^{**})$ are assigned the prices $(L,L)$ for sure.

Therefore, the only difference is the location of the threshold types: $x^*$ in Theorem \ref{thm: ic+ob} and $x^{**}$ in Proposition \ref{prop: ic}. This difference is due to the retailers' obedience constraints in Theorem \ref{thm: ic+ob} (and the lack thereof in Proposition \ref{prop: ic}).
To explain this, note first that $x^{**}\leq x^*$. Let us focus on the case with $x^{**}<x^*$; otherwise the two mechanisms are the same. If the obedience constraints were absent, the broker would prefer a larger region of $x$ at which $(H,L)$ is offered, because introducing more inefficiency would reduce the information rent more. However, with the obedience constraints, that would violate $S_1$'s obedience constraint: The retailer $S_1$, conditional on receiving the price recommendation $H$, perceives that it is either because $x<\ul{x}$ and hence the price offer is $(H,H)$, or because $x>x^{**}$ and hence the price offer is $(H,L)$. If in the first case, it is better for him to obey $H$; while if in the second case, it is better for him to deviate and offer $L$. $x^*$ is the threshold such that those two effects cancel out each other, and hence, the obedience is satisfied in expectation. However, with $x^{**}<x^*$, the second effect dominates, and hence $S_1$ would deviate to $L$; the obedience constraint fails.

As a consequence, the retailers' obedience constraints imply (i) more efficient allocations, (ii) higher information rents for the consumers, and (iii) a smaller broker revenue.

\subsection{Welfare effects of privacy protection}

Let us compare the welfare effects of introducing privacy protection (and the consumers' truth-telling constraints as a result) under these two market structures.

First, consider the case of duopoly retailers as analyzed in Sections \ref{sec: ob} and \ref{sec: ic+ob}. Without privacy protection, both retailers offer the same price to the consumers, and hence they always buy from $S_1$. In this sense, the allocation is always efficient. With privacy protection, inefficiency is introduced for the consumers with $x>x^{**}$.

For the case of the broker's price determination, without privacy protection, it is obvious that the broker extracts the entire surplus of the market, and this surplus is maximized by the fully efficient allocation (that is, the consumers all buy from $S_1$). With privacy protection, inefficiency is introduced in the same way as in the duopoly case but with a higher threshold $x^*\geq x^{**}$ (i.e., smaller inefficiency). 

To conclude, the welfare effect of introducing privacy protection is two-fold:
inefficiency for those who are closer to be indifferent; and information rents for those who are far from indifferent. Regarding the inefficiency, the effect of introducing privacy protection would be less drastic if the market structure is more competitive. Conversely, regarding the information rent, the effect would be more drastic under a more competitive market structure.

\section{Discussions}
We now aim to discuss some open questions left for further investigation of the information broker problem. All the proofs are relegated to Online Appendix.

\paragraph{Further Considerations on  Non-uniform Assumptions in Asymmetric Hotelling.}

The uniform assumption employed in Theorem \ref{thm: ic+ob} greatly simplifies our analysis. While the uniform assumption is usually seen in the analysis of a duopoly market, e.g., \cite{bounie2021selling},  we attempt to drop the uniform distribution in this section and consider a more general setup here. 
Unfortunately, the optimal information structure turns out to be very complicated. 

While the $\overline{x}$-IR always binds in the optimal solution of Theorem \ref{thm: ic+ob},  Example \ref{exampleoneofirbindsonldy} shows that there exist cases in general settings where  $\overline{x}$-IR may not bind.
\begin{example}\label{exampleoneofirbindsonldy}
    Consider three consumers $x_1 = \frac{3}{8}$, $x_2=\frac{4}{6}$ and  $x_3=\frac{5}{6}$, which have probabilities as $\lambda(x_1)=0.9$, $\lambda(x_2)=0.05$ and  $\lambda(x_3)=0.05$. Let the prices be $H= 10$, $L=9$ and the transportation fee be $t=1$. The utility of the product is $V=1000$. The optimization can be formulated as a linear program. In the optimal solution, we can observe that the $x_2$-IR binds, but the $x_3$-IR does not bind. Indeed, by Figure \ref{fig:brokergraphd}, such cases usually happen when $y_2 < \frac{1}{2}$ and $y_2=y_3$. Indeed, by solving the linear program, we can have $y_2=y_3\approx 0.234$ in the optimal design.
\end{example}
In Online Appendix \ref{threeconsumer-opt-characterization}, we completely characterize the optimal structure for a special case with $3$ consumers where the distribution could be arbitrary.
While the structural characterization is similar to that we have for Theorem \ref{thm: ic+ob}, one of the critical differences is that in this general setup, the optimal design does not have $\pi((L,L)|x)\ge \frac{1}{2}, \forall x \in [\underline{x}, \overline{x}]$ as in Theorem \ref{thm: ic+ob}. We believe that it would be an interesting open question to find the optimal information structure in a general setting without uniform assumptions.

\begin{figure}
    \centering
\includegraphics[width=0.3\textwidth]{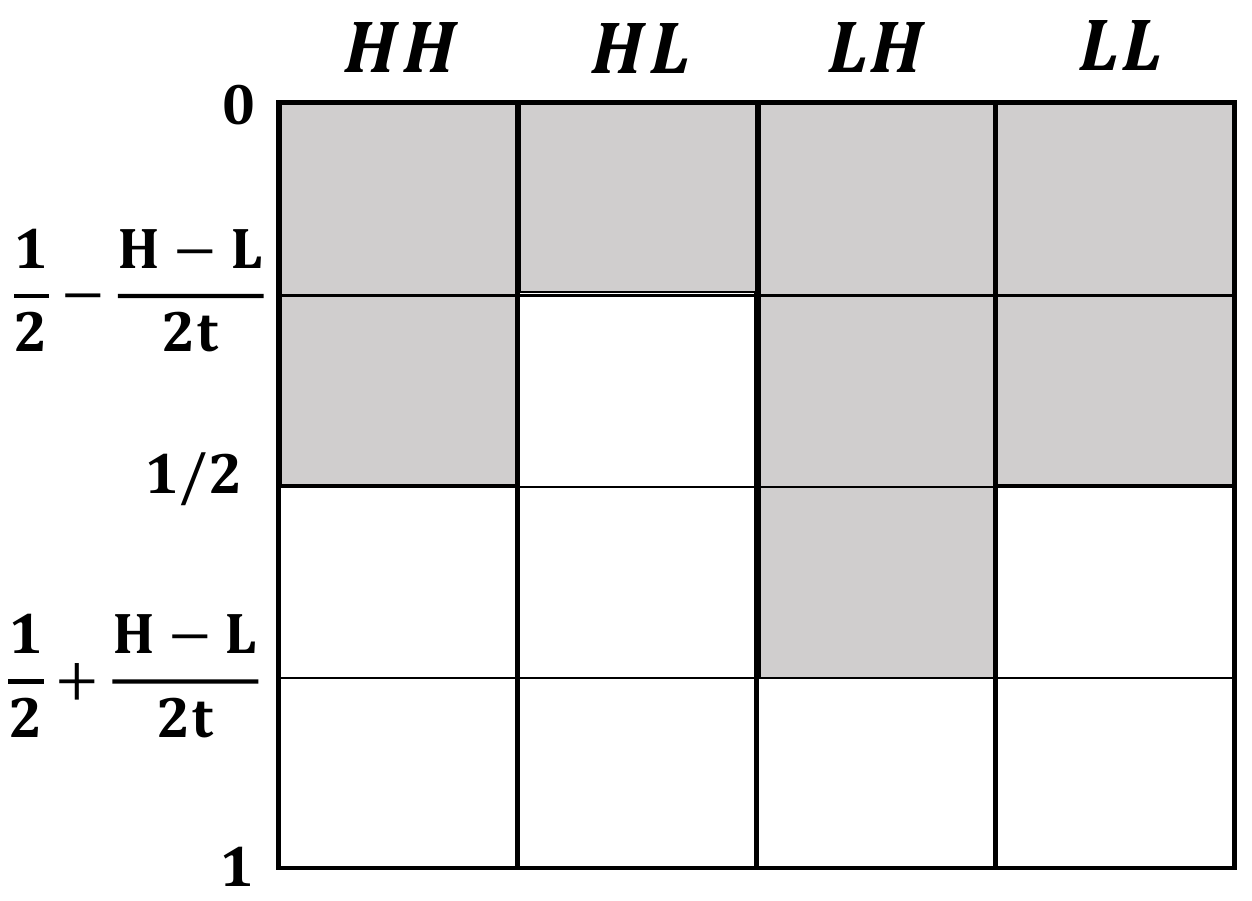}
    \caption{There are three indifferent locations that divide the market.}
    \label{signal_region_sym}
\end{figure}

\paragraph{Further Considerations on Information Design in Symmetric Hotelling.}
Another interesting open question is to explore the symmetric Hotelling where the values are $V_1=V_2=V$. The market segmentation is depicted in Figure \ref{signal_region_sym}. 
 Given the nice structure in asymmetric Hotelling, i.e., the broker only sends the signal $(H, H)$ if consumers $x\in [0, 1-\frac{H-L}{2t}]$ signals $(H, L), (L, L)$ if  consumers $x\in [1-\frac{H-L}{2t}, 1]$, it is expected that similar structures may exist in the symmetric Hotelling. 
 Unfortunately, even in the simplest setup with only two consumers, it could be challenging to characterize the optimal structure. 

In the following, consider two consumers $x_1 \in [0, \frac{1}{2})$ and $x_2 \in ( \frac{1}{2}, 1]$. The probability masses are $\lambda(x_1)$ and $\lambda(x_2)$.

\begin{example} \label{noinformatondtructionasyymeteriaehooteling}
The sufficiently large value of products is $V= 10000$. Let prices $H=10$ and $L=1$. The transportation fee is $t =18$. 
     We can see that the market is segmented into four segments $[0, \frac{1}{4}]$, $[ \frac{1}{4}, \frac{1}{2}]$, $[\frac{1}{2},  \frac{3}{4}]$ and $[\frac{3}{4}, 1]$. Consider two consumers $x_1 = \frac{3}{8}$ and $x_2 = \frac{9}{16}$, belonging to the second and third segment respectively. The probability masses are $\lambda(x_1) = 0.95$ and $\lambda(x_2)=0.05$.  The optimal information structure  is solved in Table \ref{tab:my_label}. 

\begin{table*}[h]
    \centering
        \caption{Optimal information structure.}
        %\vspace{-0.2cm}
\begin{tabular}{|c|c|c|c|c|}
\hline \diagbox{consumers}{signals}&($H$, $H$)&($H$, $L$)&($L$, $H$)&($L$, $L$)\\
\hline 
(NULL)&$0$&$0$&$0$&$0$\\
\hline $x_1$&$0.00854$& $0.07595$ &$ 0.05831$& $0.80719$\\
\hline $x_2$ & $0.00095$ & $0.0125$ & $0$ & $0.03655$\\
\hline
(NULL)&$0$&$0$&$0$&$0$\\
\hline 
\end{tabular}
    \label{tab:my_label}
    %\vspace{-0.2cm}
\end{table*}
\end{example}
 Note that the Asymmetric Hotelling can be considered as a single-parameter mechanism design problem in the sense that we only need to determine $\pi((L, L)|x)$ for all the consumers in the second segment. However, different from the results in asymmetric Hotelling, Example \ref{noinformatondtructionasyymeteriaehooteling} shows that the optimal solution may not exhibit such a similar structure in symmetric Hotelling and thus, the optimal solution is challenging to find.
We suspect this is because the problem may be intrinsically a multi-parameter mechanism design problem whose optimal solution is believed to be complicated. As we see in Example \ref{noinformatondtructionasyymeteriaehooteling}, two consumers $x_1$ and $x_2$ can behave differently, e.g., under signal $(H, H)$, $x_1$ prefers to purchase from seller $S_1$ while $x_1$ prefers to purchase from seller $S_2$. This difference is crucial, as  it now allows obedience constraints to  play a more active role in designing the optimal information structure, instead of merely determining a threshold, as in the asymmetric Hotelling case. To see this, consider setups similar to Example \ref{noinformatondtructionasyymeteriaehooteling} with two consumers located at the second and third segments, respectively. Without loss of generality, assume $x_1 \le 1-x_2$. 
\begin{proposition}\label{odefneioasymmeteric}
    Let $x  \triangleq \pi(H, L|x_2)$ and $y \triangleq \pi(L,H|x_1)$. If without the obedience constraints, there exists one optimal design having a structure where $\pi(L,H|x_1)>0$, $\pi(L,L|x_1)>0$, $\pi(L,L|x_2)>0$ and $\pi(H,L|x_2)>0$, while all other entries are $0$. Specifically, $x\ge \frac{1-x_1-x_2}{1-2x_1}$ and $y\ge 0$.
\end{proposition}
Comparing the solution in Example \ref{noinformatondtructionasyymeteriaehooteling} and Proposition \ref{odefneioasymmeteric}, it can be seen that the optimal information structure is significantly changed due to the presence of obedience constraints. Without the obedience constraints, the broker achieves the maximum possible revenue for this simple case.

%\section{AB}

%\input{no_uniform_bk2}
%\input{symmetric}

\section{Proof of Theorem \ref{thm: ic+ob}}\label{sec:opt-one-side}

We will present our full proof for Theorem \ref{thm: ic+ob}. First, we characterize the general structure of optimal design in Section \ref{subsection:structure}. Then, in Section \ref{subsection:threosdhold}, we prove the threshold-based mechanism and show how to calculate the optimal threshold. Theorem \ref{thm: ic+ob} is restated in a more technical way as follows.

\begin{theorem}\label{opt_theorem_info_design}
    There exists an optimal information structure such that
    \begin{itemize}
        \item For consumers $x\in [0, 1-\frac{H-L}{2t}]$, $\pi(H, H|x) = 1$;
        \item For consumers $x\in [1-\frac{H-L}{2t}, 1]$, there exists some threshold $x^*$ such that 
        \begin{equation}\label{optimalstructureseseg}
            \left\{
\begin{aligned}
\pi(L, L|x)  = 1, &\quad  x \le x^*\\\
\pi(L, L|x) =\pi(H,L|x) = \frac{1}{2}, &\quad  x >x^*
\end{aligned}
\right.
\end{equation}
%where $x^*$ will be defined to satisfy the obedience constraints.
    \end{itemize}
\end{theorem}

%\newpage
%\section{AB}
%To be more specific, we consider the consumer $x$ whose value for products of seller $S_1$ is $V-tx$ and that of seller $S_2$ is $V-t(1-x)$. We know $V-tx \ge V-t(1-x)$ for $x \in [0, \frac{1}{2}]$, which implies that the quality of $S_1$ products is better.

%Following the existing literature (e.g., \cite{bounie2021selling}), we assume $\Lambda$ to be a uniform distribution over $[0, \frac{1}{2}]$. 

%\jj{==========}
%\jj{Add a paragraph of interpreting Theorem \ref{opt_theorem_info_design}.}

\subsection{Structural Characterizations of Optimal Design}\label{subsection:structure}
It is not hard to know that in the optimal solution, the transfer $m_i$ equals the  expected revenue that the seller $S_i$  makes, i.e., the constraints $(\ref{sellerindrational_con})$ bind. 
%This means that the seller makes no profit from the market if the sellers lose all their private information to the platforms.
The following lemma characterizes the structure of the optimal design $\pi$.
\begin{lemma}\label{lemmasparsestructure}
    In the optimal information structure $\pi$, the platform only sends signal $(H,H)$ for $x\in [0, 1-\frac{H-L}{2t}]$, and sends both signals $(H,L)$ and $(L, L)$ for $x\in [1-\frac{H-L}{2t}, 1]$.
\end{lemma}
\begin{proof}
    First, we note that it is optimal to set $\pi(H,H|x) = 0$ for all $x\in [1-\frac{H-L}{2t}, 1]$. By the obedience constraint for seller $S_2$ under recommended price $H$, if $\pi(H,H|x) > 0$ for $x\in [1-\frac{H-L}{2t}, 1]$, then $S_2$ will benefit from deviating to price $L$ when receiving $H$ signal.

    Next, we show that there exists one optimal solution where $\pi(H, L|x) = 0$ for all $x \in [0, 1-\frac{H-L}{2t}]$. 
Consider one solution $\pi(H, L|x) > 0$. By moving the probability from $\pi(H,L|x)$ to $\pi(H,H|x)$, the consumer $x$ still purchases from $S_1$ at a price $H$ and receives the same utility. The seller's revenue does not change. Hence, it does not affect the IC constraints of $x$ misreporting to any $x' \in [0, 1]$. 
Similarly, for $x' \in [0, 1-\frac{H-L}{2t}]$, the IC constraints of $x'$ misreporting to $x$ continue to hold. To prove that, 
we only need to show that the IC constraint of $x' \in [1-\frac{H-L}{2t}, 1]$ misreporting to $ x$ still holds.
This is because under the probability mass $\pi(H,L|x)$, consumer $x'$ originally gets utility $V_2-(1-x')t-L$ when misreporting, while after modification, consumer $x'$ gets utility $V_1-x't -H$. We have 
\[
V_2-(1-x')t-L \ge V_1-x't-H
\]
where the inequality holds due to $x'\ge 1-\frac{H-L}{2t}$. 
Similarly, we can show that moving probability mass from $\pi(L, H|x)$  to $\pi(L,L|x)$ for all $x \in [0, 1]$ does not change the feasibility.

        Lastly, we have $\pi(L, L|x) =0$ for all $x\in [0, 1-\frac{H-L}{2t}]$. We first notice that moving all probability from $\pi(L, L|x)$ to $\pi(H, H|x)$ for all $x\in [0, 1-\frac{H-L}{2t}]$ does not change the obedience constraints. Then, we show that the IC constraints continue to hold. 
Consider  moving probability from $\pi(L, L|x)$ to $\pi(H, H|x)$ for some $x\in [0, 1-\frac{H-L}{2t}]$, the utility decreases by 
    \[
    \Delta U(x) =H-L
    \]
    To maintain the payoff of consumer $x$ unchanged,  we decrease the payment $m_c(x)$ by the $\Delta U(x)$. 
    Note that the utility of any $x' \in [0, 1]$ misreporting $x$ also decreases by $\Delta U(x)$.
    Hence, the feasibility of IC constraints remains. The broker's revenue does not change, as seller $S_1$ makes additional $\Delta U(x)$ revenue from the modification and the platform can increase $m_1$ by the same amount. 
\end{proof}

Since all consumers $x\in [0, 1-\frac{H-L}{2t}]$ send only signal $(H,H)$, we can quickly have the following corollary by the IC constraints among the consumers in the first segment.

\begin{corollary}\label{lemmasparsestructureconfolerally}
    Consumers $x\in [0, 1-\frac{H-L}{2t}]$ have the same payment $m_c(x)$.
\end{corollary}

%\subsection{Simplification of Program (\ref{prog_broker})}

By Lemma \ref{lemmasparsestructure} and Corollary \ref{lemmasparsestructureconfolerally},  we equivalently assume that there is only one consumer $x_1=1-\frac{H-L}{2t}$ with probability mass $\Lambda_1 \triangleq \lambda(x_1) = 1-\frac{H-L}{2t}$ in the first segment $[0, 1-\frac{H-L}{2t}]$. This is because the utilities of consumers $x < 1-\frac{H-L}{2t}$ are  higher than that of $x_1$, which implies that any constraints that consumer $x_1$ satisfy will be satisfied by consumers $x < 1-\frac{H-L}{2t}$.

 By Lemma \ref{lemmasparsestructure}, all the consumers $x\in [0, 1-\frac{H-L}{2t}]$ buy products  at high price $H$ while all the consumers $x\in [1-\frac{H-L}{2t}, 1]$ buy products at low price $L$. Hence, Program (\ref{prog_broker}) is simplified to 
\begin{equation}\label{prog_broker_simplified}
\begin{aligned}
    \max_{\pi, m_c} \quad & \Lambda_1 H+ (1-\Lambda_1)L + \int_{0}^{1} \lambda(x) m_c(x) dx\\
    \textnormal{subject. to} \quad & (\ref{obedience_con}),  (\ref{icenticomp_con}), (\ref{irconconsumer})
\end{aligned}
\end{equation}
For consumer $x\in [1- \frac{H-L}{2t}, 1]$, define the probability $y(x) \triangleq \pi(L, L|x)$. In the remaining section, we show some properties in the optimal solution of Program (\ref{prog_broker}) which helps further simplify the problem. 
The next lemma shows that in the optimal solution, $y(x)$ is monotonically decreasing.
\begin{lemma}\label{monotone_lemma}
    In the optimal design, $y(x)\ge y(x')$ and $m_c(x) \ge m_c(x')$ for any $x\le x'$. 
\end{lemma}
\begin{proof}
By the IC constraints between $x$ and $x'$ with $x \le x'$, we have
    \begin{align}
      &y(x)[V-xt-L]+(1-y(x))[(V-t)-(1-x)t-L]-m_c(x) \nonumber \\
        &\quad\quad \ge y(x')[V-xt-L]+(1-y(x'))[(V-t)-(1-x)t-L]-m_c(x') \nonumber\\
      &y(x')[V-x't-L]+(1-y(x'))[(V-t)-(1-x')t-L]-m_c(x') \ \label{icxpritoxlemma444kn}\\
        &\quad\quad \ge y(x)[V-x't-L]+(1-y(x))[(V-t)-(1-x')t-L]-m_c(x)  \nonumber
    \end{align}
By simple calculations, the above implies that  
\[
(2y(x')-1)(x'-x)t \le (2y(x)-1)(x'-x)t
\]
Hence, we have $y(x)\ge y(x')$. Furthermore, by IC constraint (\ref{icxpritoxlemma444kn}), we have 
\[
m_c(x) -m_c(x') \ge 2(y(x)-y(x'))(1-x')t
\]
which then implies $m_c(x) \ge m_c(x')$.
\end{proof}

By IC constraints between any two consumers $x, x'$ where $x \le x'$, we can derive a sandwich inequality 
\begin{equation}\label{sandwichinequality}
    2(y(x)-y(x'))(1-x)t\ge m_c(x)-m_c(x') \ge 2(y(x)-y(x'))(1-x')t
\end{equation}
Define for the second segment where the smallest location is  $\underline{x} \triangleq 1-\frac{H-L}{2t}$ and the largest location is $\bar{x}\triangleq 1$. Note that we define $\underline{x}$ to emphasize the location in the second market segment while $x_1$ is in the first segment, though $\underline{x} = x_1$. The following lemma shows that Program (\ref{prog_broker_simplified})  can be simplified by removing many IC constraints. %\jjr{Note that the following lemma even holds when $\underline{x} > \frac{1}{2}-\frac{H-L}{2t}$. }
\begin{lemma}\label{lemma45x1toxixitox1icoptimald}
    The IC constraints $x_1 \to x_i$-IC and $x_i \to x_1$-IC  do not necessarily bind in an optimal solution, where $x_i \in (1-\frac{H-L}{2t}, 1]$.
\end{lemma}
\begin{proof}
We first show that we can remove the IC constraints $x_1 \to x_i$-IC for any $x_i \in (1-\frac{H-L}{2t}, 1]$. 
It is proved by showing that the utility of $x_1$ misreporting to $\underline{x}$ is greater than that of misreporting to $x_i$, i.e., 
\[y(\underline{x})[V_1-x_1t-L] + (1-y(\underline{x})) [V_1-x_1t-H]-m_c(\underline{x}) \ge y(x_i)[V_1-x_1t-L] + (1-y(x_i)) [V_1-x_1t-H]-m_c(x_i)\]
The above inequality is  proved  by contradiction. Suppose that the above inequality does not hold, by which we have 
\[
(y(\underline{x})-y(x_i))(H-L)< m_c(\underline{x})-m_c(x_i)
\]
However, by (\ref{sandwichinequality}), the upper bound of  $m_c(\underline{x})-m_c(x_i)$ is $(y(\underline{x})-y(x_i))(1-2\underline{x})t$. Notice that 
\begin{align*}
 2(y(\underline{x})-y(x_i))(1-\underline{x})t \le (y(\underline{x})-y(x_i))(H-L)
\end{align*}
which holds due to $2(1-\underline{x})t \le (H-L) $ and $1 - \frac{H-L}{2t} \le \underline{x}$. 
This implies that (\ref{sandwichinequality}) is violated. Hence, $x_1 \to x_i$-IC constraints do not bind in the optimal solution.

Next, we show that  $x_i\to x_1$-IC constraints do not bind in the optimal solution.
By the above arguments, we know that for consumer $x_1$, only $x_1$-IR and $x_1 \to \underline{x}$-IC constraints are left in Program (\ref{prog_broker_simplified}). 
Hence, in the optimal solution, either $x_1$-IR or $x_1 \to \underline{x}$-IC constraint binds.

If the $x_1$-IR constraint binds, i.e., 
\[
V_1-x_1 t-H -m_c(x_1) =0, 
\]
then $x_i\to x_1$-IC constraint holds since the utility of $x_i$ misreporting to $x_1$ is negative due to $x_i$ paying more transportation fees, i.e, 
\[
V_1-x_i t-H -m_c(x_1)  < V_1-x_1 t-H -m_c(x_1) =0, 
\]

If $x_1 \to \underline{x}$-IC constraint binds, i.e., 
\begin{equation}\label{eqnx1toxunlinderabind}
V_1- x_1 t -H -m_c(x_1) = y(\underline{x})[V_1-x_1t-L]+(1-y(\underline{x}))[V_1-x_1t-H]-m_c(\underline{x})
\end{equation}
then the $x_i \to x_1$-IC constraint is,
\[
y(x_i)[V_1-x_it-L]+(1-y(x_i))[V_2-(1-x_i)t-L]-m_c(x_i) \ge V_1 - x_it -H - m_c(x_1)
\]
which by (\ref{eqnx1toxunlinderabind}) implies that 
\[
y(x_i)[V_1-x_it-L]+(1-y(x_i))[V_2-(1-x_i)t-L]-m_c(x_i) \ge V_1- y(\underline{x}) L -(1-y(\underline{x}))H -x_it - m_c(\underline{x})
\]
Equivalently, we have  $x_i \to x_1$-IC constraint as 
\[
V -L -y(x_i)x_it - (1-y(x_i))(1-x_i)t -m_c(x_i) - (1-y(x_i))t\ge V-y(\underline{x})L -(1-y(\underline{x}))H-x_it-m_c(\underline{x}) 
\]
and then,
\[
m_c(\underline{x}) - m_c(x_i) \ge -y(\underline{x})L -(1-y(\underline{x}))H-x_it + L +y(x_i)x_it+(1-y(x_i))(1-x_i)t+  (1-y(x_i))t
\]
By some calculation, we finally have $x_i \to x_1$-IC constraint as 
\[
m_c(\underline{x}) - m_c(x_i) \ge (1-y(\underline{x}))(L-H)+ 2(1-y(x_i))(1-x_i)t 
\]

\begin{comment}
\begin{align*}
&y(x_i)[V-x_it-L]+(1-y(x_i))[V-(1-x_i)t-L]-m_c(x_i) \ge V-y_2L -(1-y(\underline{x}))H -x_it - m_c(\underline{x}) \\
\Leftrightarrow {} & V -L -y(x_i)x_it - (1-y(x_i))(1-x_i)t -m_c(x_i)\ge V-y(\underline{x})L -(1-y(\underline{x}))H-x_it-m_c(\underline{x}) \\
\Leftrightarrow {} & m_c(\underline{x}) - m_c(x_i) \ge -y(\underline{x})L -(1-y(\underline{x}))H-x_it + L +y(x_i)x_it+(1-y(x_i))(1-x_i)t \\
\Leftrightarrow {} & m_c(\underline{x}) - m_c(x_i) \ge (1-y(\underline{x}))(L-H)- (1- y(x_i) )x_it + (1-y(x_i))(1-x_i)t \\
\Leftrightarrow {} & m_c(\underline{x}) - m_c(x_i) \ge (1-y(\underline{x}))(L-H)+ (1-y(x_i))(1-2x_i)t 
\end{align*}
\end{comment}

By the IC constraints between $\underline{x}$ and $x_i$ and (\ref{sandwichinequality}),  we have that the lower bound of $m_c(\underline{x}) - m_c(x_i)$ is larger, i.e., 
\[
2(y(\underline{x})-y(x_i))(1-x_i)t \ge (1-y(\underline{x}))(L-H)+ 2(1-y(x_i))(1-x_i)t 
\]
which holds due to 
\begin{gather*}
    2(y(\underline{x}) - 1)(1-x_i)t \ge (1-y(\underline{x}))(L-H) \\
     x_i \ge 1-\frac{H-L}{2t} 
\end{gather*}
%then we have $x_i\to \underline{x}$-IC constraint removed. It holds because
%\begin{align*}
%    &(y(\underline{x}) - 1)(1-2x_i)t \ge (1-y(\underline{x}))(L-H) \\
%    & \frac{1}{2}-\frac{H-L}{2t} \le x_i
%\end{align*}
Hence,  $x_i\to x_1$-IC constraint is implied by the IC constraints between $x_i$ and $\underline{x}$. We have that $x_i\to x_1$-IC constraint do not necessarily bind.
\end{proof}

Lemma \ref{lemma45x1toxixitox1icoptimald} implies that the first and second segments are connected by only the IC constraints between $x_1$ and $\underline{x}$. By the proof of Lemma \ref{lemma45x1toxixitox1icoptimald}, we can see that as long as 1) the IR constraints are satisfied, and 2) the IC constraints among $x\in [\underline{x}, \bar{x}]$ and the IC constraints between $x_1$ and $\underline{x}$ are satisfied, the solution is feasible. Next, we show that the IC constraint between $x_1$ and $\underline{x}$ binds in the optimal solution. Note that to achieve Lemma \ref{x1tounderlinexic} is rather easy, as we can simply adjust the payment $m_c(x_1)$ to make the equality hold. 

%Indeed, Lemma \ref{x1tounderlinexic} holds more generally for $\underline{x} \ge \frac{1}{2}-\frac{H-L}{2t}$.

\begin{lemma}\label{x1tounderlinexic}
    $x_1 \to \underline{x}$-IC constraint binds in the optimal design, i.e., 
    \[
    V_1-x_1t-H -m_c(x_1) = y(\underline{x}) [V_1-x_1t-L] + (1-y(\underline{x})) [V_1-x_1t-H] -m_c(\underline{x})
    \]
\end{lemma}

\begin{proof}
 Let 
    \[
    U(x_1\to \underline{x}) = y(\underline{x})[V_1-x_1 t- L] + (1-y(\underline{x})) [V_1-x_1 t- H]\]
    and 
    \[
    U(\underline{x}) = y(\underline{x})[V_1-\underline{x} t- L] + (1-y(\underline{x})) [V_2-(1-\underline{x}) t- L].\]  
    We can see that $U(x_1\to \underline{x}) = U(\underline{x})$ holds. By the IR constraint, we have \[
    U(x_1\to \underline{x}) - m_c(\underline{x})\ge 
    U(\underline{x}) -m_c(\underline{x}) \ge 0,\]
 which means that to maximize the broker's revenue, $x_1 \to \underline{x}$-IC constraint must bind. 
\end{proof}

Lemma \ref{x1tounderlinexic}  implies that  since $\underline{x}=x_1$ is an indifferent location, the consumer at $\underline{x}=x_1$ following either the signal distribution for the first segment $[0, 1-\frac{H-L}{2t}]$ or the signal distribution for the second segment $[1-\frac{H-L}{2t}, 1]$ receives the same utility. This is because 
\begin{align*}
U(x_1) -m_c(x_1) & =
    y(\underline{x})  [V_1-x_1t-L] + (1-y(\underline{x})) [V_1-x_1t-H] -m_c(\underline{x}) \\
    & = y(\underline{x}) [V_1-\underline{x}t-L] + (1-y(\underline{x})) [V_2-(1-\underline{x})t-L] -m_c(\underline{x}) \\
    & = U(\underline{x}) - m_c(\underline{x})
\end{align*}
Hence, treating $x_1$ and $\underline{x}$ in the asymmetric Hotelling does not change the utility for the consumer at $x_1 = \underline{x}$ in an optimal solution. Since the measure for the point $x_1 = \underline{x}$ is {\it zero}, the broker's revenue will not be affected.

Note that when $y(x) \ge \frac{1}{2}$ for $x \in [\underline{x}, \bar{x}]$, $y(x)$ must be an upper semi-continuous function, as for the same consumer $x$, larger $y(x)$ gives higher consumer utility which implies higher broker's revenue. Indeed, given the design $y(x)$ for $x \in [\underline{x}, \bar{x}]$, IR constraints must bind at some $\tilde{x} \in [\underline{x}, \bar{x}]$ where $\lim_{x \to \tilde{x}^-} y(x)\ge \frac{1}{2}$ and  $\lim_{x \to \tilde{x}^+} y(x) \le \frac{1}{2}$. If there are a finite number of consumers, the IC constraints for two consecutive consumers $x_i$ and $x_{i+1}$ must bind at the optimal design.

\begin{lemma}\label{lemamamdyxgerate1onhoefa}
    In the optimal design, $y(x) \ge \frac{1}{2}$ for all  $x\in [\underline{x}, \bar{x}]$.
\end{lemma}
\begin{proof}
First, we show that $y(\underline{x})\ge \frac{1}{2}$ holds. Suppose there exists one solution with $y(\underline{x})<\frac{1}{2}$. Note that in this case, for all $x$, we can have $y(x)=y(\underline{x})$ and $m_c(x)=m_c(\underline{x})$. Otherwise, we can construct a new solution by increasing $y(x)$ to $y(\underline{x})$ and set $m_c(x) = m_c(\underline{x})$. With this modification, we know that all the IC and IR constraints still hold but the broker's revenue increases.

     Next, we will construct a new solution that gives larger broker revenue. By increasing all $y(x)$ to $\frac{1}{2}$ and resetting the payment, we obtain the new utility and payment $\hat{U}(\underline{x})$, $\hat{m}_c(\underline{x}) = V- L -t$ such that 
    \[
     \hat{U}(\underline{x}) - \hat{m}_c(\underline{x})=0
    \]
    Obviously, the $x_1\to \underline{x}$-IC constraint is binding. Since the payments are increased, the broker gains more revenue. Therefore, in the optimal solution, $y(\underline{x}) \ge \frac{1}{2}$.

    We assume there exists some $\tilde{x} \in [\underline{x}, \bar{x}]$ such that for all $x \le \tilde{x}$, $y(x) \ge \frac{1}{2}$, and for all $x > \tilde{x}$, $y(x) <\frac{1}{2}$. First, we show that consumers' payoff $U(y(x), x) - m_c(x)$ in the optimal solution  increases as $x \in [\underline{x}, \tilde{x}]$ decreases. Given any $x' < x \le \tilde{x}$, we have 
     \[
    U(y(x'), x') - m_c(x') \ge U(y(x), x') - m_c(x) \ge U(y(x), x) - m_c(x)
    \]   
    where the first inequality is by IC constraints and the second inequality is by  $U(y(x), x') \ge U(y(x), x)$. Similarly, we can show that the payoff decreases as $x \in [\tilde{x}, \bar{x}]$ decreases. Hence, if $\tilde{x}$ is such that $y(\tilde{x})=\frac{1}{2}$, one can simply increase all the payments by $U(y(\tilde{x}), \tilde{x}) - m_c(\tilde{x})$, where the new solution is still feasible but gives higher revenue to the broker. Therefore, in the following, we discuss the case where $y(\tilde{x}) >\frac{1}{2}$.

    Next, we will show that in the optimal solution, 
    \begin{gather*}
        \lim_{x \to \tilde{x}^-} U(y(x), x) - m_c(x)=0\\
        \lim_{x \to \tilde{x}^+}  U(y(x), x) - m_c(x)=0
    \end{gather*}

By contradiction, we assume $d_+ = \lim_{x \to \tilde{x}^+}  U(y(x), x) - m_c(x)\ge 0$ and $d_{-} = \lim_{x \to \tilde{x}^-} U(y(x), x) - m_c(x) \ge 0$. 
We discuss the case $d_+ > d_{-}$, and the discussion for the case of $d_+ \le d_{-}$ is similar.

Hence, we increase the payments $m_c(x)$ for all $x$ by $d_{-}$. This does not change IC and IR constraints. Now, we show that for any $x > \tilde{x}$, the utility of $x$ misreporting to $x' \le \tilde{x}$ is non-positive. 
In fact, since the IR constraint for $\tilde{x}$ is binding, we only need to show the payoff of consumer $x$ misreporting to $\tilde{x}$ is non-positive. This is because 
\[
U(y(\tilde{x}), x) -m_c(\tilde{x}) \le U(y(\tilde{x}), \tilde{x}) -m_c(\tilde{x}) =0 
\]
where $U(y(\tilde{x}), x) < U(y(\tilde{x}), \tilde{x})$ since $y(\tilde{x}) \ge \frac{1}{2}$. 
We also show that the utility of $x$ misreporting to $x_1$ is non-positive, i.e.
\[
V_1 - xt - H -m_c(x_1) \le 0
\]
Since we know that utility of $x$ misreporting to $\underline{x}$ is non-positive, i.e., 
\[
y(\underline{x})[V_1-xt - L] + (1-y(\underline{x})) [V_2-(1-x)t-L] - m_c(\underline{x}) \le 0,
\]
and the binding constraint between $x_1$ and $\underline{x}$, we have
\begin{align*}
    &(V_1 - xt - H) -m_c(x_1)\\
    &= (V_1 - xt - H) - \Big\{ (V_1 - x_1t - H) - \Big( y(\underline{x}) [V_1-\underline{x}t-L] + (1-y(\underline{x})) [V_2-(1-\underline{x})t-L] -m_c(\underline{x}) \Big) \Big\} \\
    & = (\underline{x}-x)t + \Big( y(\underline{x}) [V_1-\underline{x}t-L] + (1-y(\underline{x})) [V_2-(1-\underline{x})t-L] -m_c(\underline{x}) \Big) \\
    & = y(\underline{x}) [V_1-xt-L] + (1-y(\underline{x})) [V_2-(1-\underline{x})t + (\underline{x}-x)t-L] -m_c(\underline{x})
    \le 0
\end{align*}
where the last inequality holds due to that 
\[
[V_2-(1-\underline{x})t + (\underline{x}-x)t-L] \le [V_2-(1-x)t-L]
\]

Then, we increase the payments $m_c(x)$ for $x > \tilde{x}$ by $d_+ - d_{-}$. Note that the IR constraints for all $x$ with $x > \tilde{x}$ still hold. Hence, the IC constraints for $x$ misreporting to $x'$ with $x' \le \tilde{x}$ holds, as the utility of misreporting is non-positive. Moreover, the IC constraints of $x'$ with $x' \le \tilde{x}$ misreporting to $x$ with $x >  \tilde{x}$ still hold as the payments $m_c(x)$ with $x >  \tilde{x}$ increase.

Now, we show that we can construct a new solution with $y(x) \ge \frac{1}{2}$ for all $x \in [\underline{x}, \bar{x}]$ such that the broker's revenue is higher. We can increase $y(x)$ to $y(x)=\frac{1}{2}$ for all $x > \tilde{x}$, and the payments are set $m_c(x) = V-L-t$. Therefore, we can see that the payoffs for all the consumers $x$ with $x > \tilde{x}$ are $0$. We can verify that all the IC and IR constraints hold. The obedience constraints are still satisfied after increasing the allocation probability $y(x)$ in the ways discussed above. Hence, the lemma is proved. 
\end{proof}

\subsection{Characterization of Optimal Threshold}\label{subsection:threosdhold}

Now, we are ready to present the proof of Theorem \ref{opt_theorem_info_design}. We  first turn to the envelope theorem to derive the formula for payments $m_c(x)$.

Consider any $x,x' \in [\underline{x}, \bar{x}]$ where $x\le x'$ and $\underline{x} = 1-\frac{H-L}{2t}, \bar{x}= 1$. By the IC constraints,
\begin{align*}
    G(x) &= U(y(x), x) -m_c(x)\\
    & = y(x)[V_1-xt-L] + (1-y(x))[V_2-(1-x)t -L] -m_c(x) \\
    &\ge y(x')[V_1-xt-L] + (1-y(x'))[V_2-(1-x)t -L] -m_c(x') \tag{By IC constraints} \\
    & = y(x')[V_1-x't-L] + y(x')x't -y(x')xt \\
    & \quad + (1-y(x'))[V_2-(1-x')t -L] +(1-y(x'))(1-x')t - (1-y(x'))(1-x)t -m_c(x') \\
    & = U(y(x'), x') - m_c(x') + y(x')(x'-x)t + (1-y(x'))(x-x')t\\
    & = G(x') + (2y(x')-1)(x'-x)t 
\end{align*}
That implies 
\begin{gather*}
(1-2y(x'))(x-x')t \le G(x)-G(x') \le (1-2y(x))(x-x')t \\
 (1-2y(x'))t \ge \frac{G(x)-G(x')}{x-x'} \ge (1-2y(x))t
\end{gather*}
Hence, we have the derivative $G'(x) = (1-2y(x))t$, which implies that 
\[G(x) = G(\underline{x}) + \int^{x}_{\underline{x}} (1-2y(z))t dz=G(\bar{x}) + \int_{x}^{\bar{x}} (2y(z)-1)t dz.\] 
Therefore, the payment for $x\in [1-\frac{H-L}{2t}, 1]$ is 
\begin{align*}
m_c(x) = y(x)[V_1-xt-L] + (1-y(x))[V_2-(1-x)t -L] - G(\bar{x}) - \int_{x}^{\bar{x}} (2y(z)-1)t dz
\end{align*}
By Lemma \ref{x1tounderlinexic} that $x_1\to \underline{x}$-IC binds in the optimal design, the broker's revenue is rewritten as
\begin{align*}
    &\int_{\underline{x}}^{\bar{x}}  m_c(x) dx + \Lambda_1 m_c(x_1) \\
    & = \int_{\underline{x}}^{\bar{x}} \Big( y(x)[V_1-xt-L] + (1-y(x))[V_2-(1-x)t -L] - G(\bar{x}) - \int_{x}^{\bar{x}} (2y(z)-1)t dz \Big) dx \\ 
    & \quad + \Lambda_1 \Big(y(\underline{x})[V_1-\underline{x}t-L] + (1-y(\underline{x}))[V_2-(1-\underline{x})t -L] - G(\bar{x}) - \int_{{\underline{x}}}^{\bar{x}} (2y(z) - 1)t dz+y(\underline{x})(-H+L)\Big)
\end{align*}
To maximize the broker's revenue, it is optimal to have that $G(\bar{x})=0$, which is indeed satisfied as discussed in the proof of Lemma \ref{lemamamdyxgerate1onhoefa}. Hence, by the expression of payments $m_c(x)$ and $m_c(x_1)$, as long as $V$ is sufficiently large, we can ensure that the $m_c(x)$ and $m_c(x)$ are non-negative.

Finally, the optimization problem boils down to finding the  function $y(x)$ with $x \in [\underline{x}, \bar{x}]$.

\begin{lemma}\label{thereexistsomestaroptimalinfo}
    There exists some threshold $x^*$ such that the optimal structure is as in (\ref{optimalstructureseseg}).
\end{lemma}

\begin{proof}
    This is proved by showing that given an arbitrary feasible solution $y(x)\ge \frac{1}{2}$ (By Lemma \ref{lemamamdyxgerate1onhoefa}),
    %but not all the $x$ have $y(x)=\frac{1}{2}$, 
    we can construct a new solution of structure as in (\ref{optimalstructureseseg}) that gives weakly higher revenue to the broker.

    The key idea of proving this lemma is to show that after moving some probabilities up for some amount and moving some probabilities down for the same amount, the newly obtained function $\hat{y}(x)$ gives higher revenue than the original function $y(x)$. As illustrated in Figure \ref{fig:threshold}, we can start from $\bar{x}$ and $\underline{x}$ respectively,  and move the probabilities $[\underline{x}, \underline{x} + \delta_{\underline{x}}]$ towards $1$ and the probabilities $[\bar{x}, \bar{x} - \delta_{\bar{x}}]$ towards $\frac{1}{2}$. Note that we ensure the total probability of moving up and down is the same due to the requirements of obedience constraints. This step implies that there must exist some threshold $x^*$ such that $y(x)=1$ for $x \in [\underline{x}, x^*]$ and $y(x)=\frac{1}{2}$ for $x \in (x^*, \bar{x}]$. 

    In the following, we divide the analysis into two cases: 1) $y(\underline{x}) = 1$ and 2) $y(\underline{x}) < 1$. In the following discussion, we can always assume there is one consumer $\hat{x}=\bar{x}$ with measure \textit{zero} that has $y(\hat{x})=\frac{1}{2}$,  payments $m_c(\hat{x}) = V-L-t$ and thus, payoff $G(\hat{x})=0$. Any other consumer $x$ misreporting to $\hat{x}$ also has payoff \textit{zero}. Hence, it is equivalent to the IR constraints.

\textbf{Case $y(\underline{x}) = 1$.} Pick any two points $x_s, x_t$ where $\underline{x}\le x_s< x_t$. 
Choose $\delta_s$ and $\delta_t$ such that $\int_{x_s}^{x_s+\delta_s} [y(x_s) - y(x)] dx = \int_{x_t-\delta_t}^{x_t} [y(x) - y(x_t)] dx$. We make the following modifications: 
\begin{itemize}
    \item For all $x \in [x_s, x_s+\delta_s]$, we increase $y(x)$ by $\Delta(x)=y(x_s)-y(x)$;
    \item For all $x \in [x_t-\delta_t, x_t]$, we decrease $y(x)$ by $\Delta(x)=y(x)-y(x_t)$. 
\end{itemize}
Denote the probability function after modifications as $\hat{y}$, and also the revenue obtained from modifications as $\hat{m}$.
Then, we have that for any constant $t>0$,
\begin{equation}\label{eqyhatminusy}
\begin{aligned}
     \int_{\underline{x}}^{\bar{x}} \hat{y}(x)t dx - \int_{\underline{x}}^{\bar{x}} y(x)t dx &= \Big( \int_{\underline{x}}^{x_s} y(z)tdz + \int_{x_s}^{x_s+\delta_s} (y(z)+\Delta(z))tdz + \int_{x_s+\delta_s}^{x_t-\delta_t} y(z)tdz  \\
    & \quad \quad +\int_{x_t-\delta_t}^{x_t} (y(z) -\Delta(z))tdz +  \int_{x_t}^{\bar{x}} y(z)tdz \Big) - \Big( \int^{\bar{x}}_{\underline{x}} y(z) t dz  \Big) \\
    & = \int_{x_s}^{x_s+\delta_s} (\Delta(z))tdz + \int_{x_t-\delta_t}^{x_t} ( -\Delta(z))tdz \\
    & = 0
\end{aligned}
\end{equation}
We first compare the revenue collected from consumers in the first segment, i.e., $\Lambda_1 m_c(x_1)$,  and by (\ref{eqyhatminusy}), we have
\begin{align*}
    \Lambda_1\hat{m}_c(x_1) - \Lambda_1m_c(x_1) =  \Lambda_1\Big(- \int_{{\underline{x}}}^{\bar{x}} (2\hat{y}(z) - 1)t dz +  \int_{{\underline{x}}}^{\bar{x}} (2y(z) - 1)t dz \Big) = 0
\end{align*}
Next, we compare the revenue collected from the second segment, i.e., $\int_{\underline{x}}^{\bar{x}}  m_c(x) dx$. First note that
\begin{align*}
    &\int_{\underline{x}}^{\bar{x}}  m_c(x) dx  = \int_{\underline{x}}^{\bar{x}} \Big( y(x)[V_1-xt-L] + (1-y(x))[V_2-(1-x)t -L] - G(\bar{x}) - \int_{x}^{\bar{x}} (2y(z)-1)t dz \Big) dx\\
    & = \int_{\underline{x}}^{\bar{x}} V-L - y(x)xt - (1-y(x))(1-x)t - (1-y(x))t-G(\bar{x}) dx - \int_{\underline{x}}^{\bar{x}} \int_{x}^{\bar{x}} (2y(z) - 1)t dz  dx \\
    & =  \Big( \int_{\underline{x}}^{\bar{x}} V-L - y(x)xt - (1-x)t +y(x)t - y(x)xt- t + y(x)t -G(\bar{x}) dx   \Big) \\
    & \hspace{5cm}+\int_{\underline{x}}^{\bar{x}} \int_{x}^{\bar{x}} t dz  dx - \int_{\underline{x}}^{\bar{x}} \int_{x}^{\bar{x}} 2y(z)t dz  dx \\
    & =  \Big( \int_{\underline{x}}^{\bar{x}} V-L  - (1-x)t +y(x)t -t + y(x)t- G(\bar{x}) dx   \Big) + \Big( \int_{\underline{x}}^{\bar{x}} - y(x)xt  - y(x)xtdx   \Big) \\
    & \quad \quad + \int_{\underline{x}}^{\bar{x}} \int_{x}^{\bar{x}} t dz  dx - \int_{\underline{x}}^{\bar{x}} \int_{x}^{\bar{x}} 2y(z)t dz  dx 
\end{align*}
By (\ref{eqyhatminusy}), we know that the value of $\int_{\underline{x}}^{\bar{x}} y(x)t dx$ does not change after modification.  Further notice that $\int_{\underline{x}}^{\bar{x}} \int_{x}^{\bar{x}} 2y(z)t dz  dx = \int_{\underline{x}}^{\bar{x}} 2y(z)t \int^{z}_{\underline{x}}dxdz = \int_{\underline{x}}^{\bar{x}} 2y(z)t (z- \underline{x})dz$. Therefore, we have 
\begin{equation}\label{eqmcsecondsegminus}
\begin{aligned}
    &\int_{\underline{x}}^{\bar{x}}  \hat{m}_c(x) dx -  \int_{\underline{x}}^{\bar{x}}  m_c(x) dx\\
    & = \Big[  - \int_{\underline{x}}^{\bar{x}} 2 \hat{y}(x)xt dx  -\int_{\underline{x}}^{\bar{x}} 2\hat{y}(x)t (x- \underline{x})dx \Big]  - \Big[ - \int_{\underline{x}}^{\bar{x}} 2 {y}(x)xt dx  -\int_{\underline{x}}^{\bar{x}} 2{y}(x)t (x- \underline{x})dx \Big] \\
    & =  - \int_{\underline{x}}^{\bar{x}} 4\hat{y}(x)xtdx + \int_{\underline{x}}^{\bar{x}} 4y(x)xtdx   \\
    & =  -4 \Big(  \int_{x_s}^{x_s+\delta_s} (y(z)+\Delta(x))xtdz  
     +\int_{x_t-\delta_t}^{x_t} (y(z) -\Delta(x))xtdz  - \int_{x_s}^{x_s+\delta_s} y(z)xtdz  
     -\int_{x_t-\delta_t}^{x_t} y(z)x tdz  \Big) \\
     &= -4 \Big(  \int_{x_s}^{x_s+\delta_s} \Delta(x)xtdz  
     -\int_{x_t-\delta_t}^{x_t} \Delta(x)xtdz  \Big) \ge 0
\end{aligned}
\end{equation}
where the last inequality holds because $x_s +\delta_s \le x_t-\delta_t$ and $\int_{x_s}^{x_s+\delta_s} \Delta(x) dx = \int_{x_t-\delta_t}^{x_t} \Delta(x) dx$ by definition.

\textbf{Case $y(\underline{x}) < 1$.} Compared with the case $y(\underline{x}) = 1$, we only need to additionally  consider the situation where we increase $y(\underline{x})$ by some $\Gamma > 0$.
For some $\delta_s>0$ and $x_s=\underline{x}$,  we increase $y(x)$ by $\Gamma$ for all $x \in [x_s, x_s+\delta_s]$. Choose $\delta_t>0$ and $x_t- \delta_t \ge x_s + \delta_s$ such that $\int_{\underline{x}}^{\underline{x}+\delta_s} \Gamma dx = \int_{x_t-\delta_t}^{x_t} [y(x) - y(x_t)] dx$. 

Denote the new probability function as $\hat{y}$. Similar to (\ref{eqyhatminusy}), we still have for any $t>0$ that 
\[\int_{\underline{x}}^{\bar{x}} \hat{y}(x)t dx - \int_{\underline{x}}^{\bar{x}} y(x)t dx = 0\]
Also, we have that the revenue collected from the first segment does not change.
\begin{align*}
    \hat{m}_c(x_1) - m_c(x_1) & = \Big(\hat{y}(\underline{x})[V_1-\underline{x}t-L] + (1-\hat{y}(\underline{x}))[V_2-(1-\underline{x})t -L] +\hat{y}(\underline{x})(-H+L)\Big) \\
    &\qquad - \Big({y}(\underline{x})[V_1-\underline{x}t-L] + (1-{y}(\underline{x}))[V_2-(1-\underline{x})t -L] +{y}(\underline{x})(-H+L)\Big)\\
    & = \Big(\Gamma[V_1-\underline{x}t-L] -\Gamma[V_2-(1-\underline{x})t -L] + \Gamma(L-H) \Big)\\
    & = 0
\end{align*}
Similar to (\ref{eqmcsecondsegminus}), we calculate the difference as
\begin{equation*}
\begin{aligned}
    &\int_{\underline{x}}^{\bar{x}}  \hat{m}_c(x) dx -  \int_{\underline{x}}^{\bar{x}}  m_c(x) dx\\
    & =  - \int_{\underline{x}}^{\bar{x}} 4\hat{y}(x)xtdx + \int_{\underline{x}}^{\bar{x}} 4y(x)xtdx   \\
    & =  -4 \Big(  \int_{\underline{x}}^{x_s+\delta_s} (y(z)+\Gamma)xtdz  
     +\int_{x_t-\delta_t}^{x_t} (y(z) -\Delta(x))xtdz  - \int_{\underline{x}}^{\underline{x}+\delta_s} y(z)xtdz  
     -\int_{x_t-\delta_t}^{x_t} y(z)x tdz  \Big) \\
     &= -4 \Big(  \int_{\underline{x}}^{\underline{x}+\delta_s} \Gamma xtdz  
     -\int_{x_t-\delta_t}^{x_t} \Delta(x)xtdz  \Big)\\
     &\ge 0
\end{aligned}
\end{equation*}
where the last inequality holds because $\underline{x} +\delta_s \le x_t-\delta_t$ and $\int_{\underline{x}}^{\underline{x}+\delta_s} \Gamma dx = \int_{x_t-\delta_t}^{x_t} \Delta(x) dx$ by definition. This implies that by the above modification, we have higher revenue for the broker. 

Finally, if it is in the case $y(\underline{x})<1$, we can first  move up the probability as discussed above so that $y(\underline{x}) =1$. Then, the problem reduces to the case $y(\underline{x}) =1$ and the analysis follows. The analysis for case $y(\underline{x}) =1$ implies that there must exist some threshold $x^*$ such that 
\[
\int_{\underline{x}}^{x^*} [1 - y(x)] dx = \int_{x^*}^{\bar{x}} [y(x) - \frac{1}{2}] dx
\]
The lemma is proved.
\end{proof}

\begin{comment}
\begin{align*}
    &\int_{\underline{x}}^{\bar{x}}  m_c(x) dx + \Lambda_1 m_c(x_1) \\
    & = \int_{\underline{x}}^{\bar{x}} \Big( y(x)[V-xt-L] + (1-y(x))[V-(1-x)t -L] - U(\bar{x}) - \int_{x}^{\bar{x}} (2y(z)-1)t dz \Big) dx \\ 
    & \quad + \Lambda_1 \Big(y(\underline{x})[V-\underline{x}t-L] + (1-y(\underline{x}))[V-(1-\underline{x})t -L] - U(\bar{x}) - \int_{{\underline{x}}}^{\bar{x}} (2y(z) - 1)t dz+y(\underline{x})(-H+L)\Big)
\end{align*}
\end{comment}

The last step of the proof is to calculate $x^*$, which is determined by the obedience constraint. 
Recall that in the optimal solution, $G(\bar{x}) = 0$. Therefore, given the threshold $x^*$, the broker's revenue from consumers is 
\begin{equation}
    \begin{aligned}
        R(x^*) & = \int_{\underline{x}}^{\bar{x}}  m_c(x) dx + \Lambda_1 m_c(x_1)\\
        &= \int_{\underline{x}}^{x^*} \Big( [V_1-xt-L]   - \int_{x}^{x^*} t dz \Big) dx  + \int_{x^*}^{\bar{x}} \Big( \frac{1}{2}[V_1-xt-L] + \frac{1}{2}[V_2-(1-x)t -L]  \Big) dx \\
        &\quad + \Lambda_1 \Big([V_1-\underline{x}t-H]   - \int_{{\underline{x}}}^{x^*} t dz \Big) \\
        & = \int_{\underline{x}}^{x^*} \Big( [V-xt-L]   - t(x^*-x) \Big) dx + (\bar{x}-x^*)(V-t-L) - \Lambda_1(x^*-\underline{x})t + C\\
        &= (x^*-\underline{x})(V-tx^*-L) + (\bar{x}-x^*)(V-t-L) -\underline{x}(x^*-\underline{x})t+C \\
        &= x^*(V-tx^*-L) - \underline{x}(V-tx^*-L) -x^*(V-t-L) -\underline{x}x^*t +C \\
        & = -t(x^*)^2 + tx^*+C
    \end{aligned}
\end{equation}
where $C$ is some constant unrelated to $x^*$.
Hence, the broker's revenue $R(x^*)$ achieves its maximum at $-\frac{t}{-2t} = \frac{1}{2}$.  

%Recall that $x^*\ge \underline{x}$. It means that without the obedience constraint, if $\underline{x} \le \frac{1}{4}$, then we let $x^*=\frac{1}{4}$; otherwise, we let $x^*=\underline{x}$. Also, note that $G(x^*)$ is a decreasing function if $x^*\ge \frac{1}{4}$.

Finally, we take the obedience constraints into consideration. Recall $\underline{x} = \Lambda_1$. We use $\pi(s, i)$ to denote the probability mass of the signaling scheme when the signal is $s$, and it is at the $i^{th}$ segment. Formally, it is defined as 
\[
\pi(s, 1) \triangleq \int_{0}^{\underline{x}} \pi(s|x) \lambda(x) dx \qquad \pi(s, 2) \triangleq \int_{\underline{x}}^{\bar{x}} \pi(s|x) \lambda(x) dx
\]
Hence, the obedience constraint for the signal $s_1 = H$ for  seller $S_1$ is equivalent to
\[
    \pi(HH, 1) H \ge (\pi(HH, 1) + \pi(HL, 2))L, 
\]
which is equivalent to 
\[
    \underline{x} H \ge (\underline{x} + (1-x^*))L
    \]
    and finally
    \[
    x^* \ge 1-\frac{\underline{x}(H-L)}{L}
\]
Therefore, the optimal $x^*$ will be
\begin{itemize}
    \item If $\underline{x} \ge \frac{1}{2}$, then $x^* =\max\{   \underline{x} , 1-\frac{\underline{x}(H-L)}{L}\}$.
    \item If $\underline{x} < \frac{1}{2}$, then if $1-\frac{\underline{x}(H-L)}{L} < \frac{1}{2}$, we let $x^*=\frac{1}{2}$, otherwise, let $x^* = 1-\frac{\underline{x}(H-L)}{L}$.
\end{itemize}

\bibliographystyle{ACM-Reference-Format}
\bibliography{bibname}


\begin{thebibliography}{19}

%%% the Web address.

\ifx \showCODEN    \undefined \def \showCODEN     #1{\unskip}     \fi
\ifx \showDOI      \undefined \def \showDOI       #1{#1}\fi
\ifx \showISBNx    \undefined \def \showISBNx     #1{\unskip}     \fi
\ifx \showISBNxiii \undefined \def \showISBNxiii  #1{\unskip}     \fi
\ifx \showISSN     \undefined \def \showISSN      #1{\unskip}     \fi
\ifx \showLCCN     \undefined \def \showLCCN      #1{\unskip}     \fi
\ifx \shownote     \undefined \def \shownote      #1{#1}          \fi
\ifx \showarticletitle \undefined \def \showarticletitle #1{#1}   \fi
\ifx \showURL      \undefined \def \showURL       {\relax}        \fi
% invisible to TeX
\providecommand\bibfield[2]{#2}
\providecommand\bibinfo[2]{#2}
\providecommand\natexlab[1]{#1}
\providecommand\showeprint[2][]{arXiv:#2}

\bibitem[\protect\citeauthoryear{Agarwal, Dahleh, and Sarkar}{Agarwal
  et~al\mbox{.}}{2019}]%
        {agarwal2019marketplace}
\bibfield{author}{\bibinfo{person}{Anish Agarwal}, \bibinfo{person}{Munther
  Dahleh}, {and} \bibinfo{person}{Tuhin Sarkar}.}
  \bibinfo{year}{2019}\natexlab{}.
\newblock \showarticletitle{A marketplace for data: An algorithmic solution}.
  In \bibinfo{booktitle}{\emph{Proceedings of the 2019 ACM Conference on
  Economics and Computation}}. \bibinfo{pages}{701--726}.
\newblock


\bibitem[\protect\citeauthoryear{Amazon}{Amazon}{2024}]%
        {Howmuchd35:online}
\bibfield{author}{\bibinfo{person}{Amazon}.} \bibinfo{year}{2024}\natexlab{}.
\newblock \bibinfo{title}{How much does it cost to sell with Amazon?}
\newblock
  \bibinfo{howpublished}{\url{https://sell.amazon.com/pricing##other-costs}}.
\newblock
\newblock
\shownote{(Accessed on 04/16/2024).}


\bibitem[\protect\citeauthoryear{Arieli and Babichenko}{Arieli and
  Babichenko}{2019}]%
        {arieli2019private}
\bibfield{author}{\bibinfo{person}{Itai Arieli} {and} \bibinfo{person}{Yakov
  Babichenko}.} \bibinfo{year}{2019}\natexlab{}.
\newblock \showarticletitle{Private bayesian persuasion}.
\newblock \bibinfo{journal}{\emph{Journal of Economic Theory}}
  \bibinfo{volume}{182} (\bibinfo{year}{2019}), \bibinfo{pages}{185--217}.
\newblock


\bibitem[\protect\citeauthoryear{Bergemann and Bonatti}{Bergemann and
  Bonatti}{2015}]%
        {bergemann2015selling}
\bibfield{author}{\bibinfo{person}{Dirk Bergemann} {and}
  \bibinfo{person}{Alessandro Bonatti}.} \bibinfo{year}{2015}\natexlab{}.
\newblock \showarticletitle{Selling cookies}.
\newblock \bibinfo{journal}{\emph{American Economic Journal: Microeconomics}}
  \bibinfo{volume}{7}, \bibinfo{number}{3} (\bibinfo{year}{2015}),
  \bibinfo{pages}{259--294}.
\newblock


\bibitem[\protect\citeauthoryear{Bergemann, Bonatti, and Smolin}{Bergemann
  et~al\mbox{.}}{2018}]%
        {bergemann2018design}
\bibfield{author}{\bibinfo{person}{Dirk Bergemann}, \bibinfo{person}{Alessandro
  Bonatti}, {and} \bibinfo{person}{Alex Smolin}.}
  \bibinfo{year}{2018}\natexlab{}.
\newblock \showarticletitle{The design and price of information}.
\newblock \bibinfo{journal}{\emph{American economic review}}
  \bibinfo{volume}{108}, \bibinfo{number}{1} (\bibinfo{year}{2018}),
  \bibinfo{pages}{1--48}.
\newblock


\bibitem[\protect\citeauthoryear{Bonatti, Dahleh, Horel, and Nouripour}{Bonatti
  et~al\mbox{.}}{2024}]%
        {BONATTI2024105779}
\bibfield{author}{\bibinfo{person}{Alessandro Bonatti},
  \bibinfo{person}{Munther Dahleh}, \bibinfo{person}{Thibaut Horel}, {and}
  \bibinfo{person}{Amir Nouripour}.} \bibinfo{year}{2024}\natexlab{}.
\newblock \showarticletitle{Selling information in competitive environments}.
\newblock \bibinfo{journal}{\emph{Journal of Economic Theory}}
  \bibinfo{volume}{216} (\bibinfo{year}{2024}), \bibinfo{pages}{105779}.
\newblock
\showISSN{0022-0531}
\urldef\tempurl%
\url{https://doi.org/10.1016/j.jet.2023.105779}
\showDOI{\tempurl}


\bibitem[\protect\citeauthoryear{Bounie, Dubus, and Waelbroeck}{Bounie
  et~al\mbox{.}}{2021}]%
        {bounie2021selling}
\bibfield{author}{\bibinfo{person}{David Bounie}, \bibinfo{person}{Antoine
  Dubus}, {and} \bibinfo{person}{Patrick Waelbroeck}.}
  \bibinfo{year}{2021}\natexlab{}.
\newblock \showarticletitle{Selling strategic information in digital
  competitive markets}.
\newblock \bibinfo{journal}{\emph{The RAND Journal of Economics}}
  \bibinfo{volume}{52}, \bibinfo{number}{2} (\bibinfo{year}{2021}),
  \bibinfo{pages}{283--313}.
\newblock


\bibitem[\protect\citeauthoryear{Chen, Li, and Xu}{Chen et~al\mbox{.}}{2022}]%
        {chen2022selling}
\bibfield{author}{\bibinfo{person}{Junjie Chen}, \bibinfo{person}{Minming Li},
  {and} \bibinfo{person}{Haifeng Xu}.} \bibinfo{year}{2022}\natexlab{}.
\newblock \showarticletitle{Selling data to a machine learner: Pricing via
  costly signaling}. In \bibinfo{booktitle}{\emph{International Conference on
  Machine Learning}}. PMLR, \bibinfo{pages}{3336--3359}.
\newblock


\bibitem[\protect\citeauthoryear{Dworczak and Martini}{Dworczak and
  Martini}{2019}]%
        {dworczak2019simple}
\bibfield{author}{\bibinfo{person}{Piotr Dworczak} {and}
  \bibinfo{person}{Giorgio Martini}.} \bibinfo{year}{2019}\natexlab{}.
\newblock \showarticletitle{The simple economics of optimal persuasion}.
\newblock \bibinfo{journal}{\emph{Journal of Political Economy}}
  \bibinfo{volume}{127}, \bibinfo{number}{5} (\bibinfo{year}{2019}),
  \bibinfo{pages}{1993--2048}.
\newblock


\bibitem[\protect\citeauthoryear{Elliott, Galeotti, Koh, and Li}{Elliott
  et~al\mbox{.}}{2021}]%
        {elliott2021market}
\bibfield{author}{\bibinfo{person}{Matthew Elliott}, \bibinfo{person}{Andrea
  Galeotti}, \bibinfo{person}{Andrew Koh}, {and} \bibinfo{person}{Wenhao Li}.}
  \bibinfo{year}{2021}\natexlab{}.
\newblock \showarticletitle{Market segmentation through information}.
\newblock \bibinfo{journal}{\emph{Available at SSRN 3432315}}
  (\bibinfo{year}{2021}).
\newblock


\bibitem[\protect\citeauthoryear{Gentzkow and Kamenica}{Gentzkow and
  Kamenica}{2016}]%
        {gentzkow2016rothschild}
\bibfield{author}{\bibinfo{person}{Matthew Gentzkow} {and}
  \bibinfo{person}{Emir Kamenica}.} \bibinfo{year}{2016}\natexlab{}.
\newblock \showarticletitle{A Rothschild-Stiglitz approach to Bayesian
  persuasion}.
\newblock \bibinfo{journal}{\emph{American Economic Review}}
  \bibinfo{volume}{106}, \bibinfo{number}{5} (\bibinfo{year}{2016}),
  \bibinfo{pages}{597--601}.
\newblock


\bibitem[\protect\citeauthoryear{Guo and Shmaya}{Guo and Shmaya}{2019}]%
        {guo2019interval}
\bibfield{author}{\bibinfo{person}{Yingni Guo} {and} \bibinfo{person}{Eran
  Shmaya}.} \bibinfo{year}{2019}\natexlab{}.
\newblock \showarticletitle{The interval structure of optimal disclosure}.
\newblock \bibinfo{journal}{\emph{Econometrica}} \bibinfo{volume}{87},
  \bibinfo{number}{2} (\bibinfo{year}{2019}), \bibinfo{pages}{653--675}.
\newblock


\bibitem[\protect\citeauthoryear{Ichihashi}{Ichihashi}{2021}]%
        {ichihashi2021competing}
\bibfield{author}{\bibinfo{person}{Shota Ichihashi}.}
  \bibinfo{year}{2021}\natexlab{}.
\newblock \showarticletitle{Competing data intermediaries}.
\newblock \bibinfo{journal}{\emph{The RAND Journal of Economics}}
  \bibinfo{volume}{52}, \bibinfo{number}{3} (\bibinfo{year}{2021}),
  \bibinfo{pages}{515--537}.
\newblock


\bibitem[\protect\citeauthoryear{Ichihashi and Smolin}{Ichihashi and
  Smolin}{2022}]%
        {ichihashi2022collection}
\bibfield{author}{\bibinfo{person}{Shota Ichihashi} {and} \bibinfo{person}{Alex
  Smolin}.} \bibinfo{year}{2022}\natexlab{}.
\newblock \showarticletitle{Data collection by an informed seller}.
\newblock  (\bibinfo{year}{2022}).
\newblock


\bibitem[\protect\citeauthoryear{Kamenica and Gentzkow}{Kamenica and
  Gentzkow}{2011}]%
        {kamenica2011bayesian}
\bibfield{author}{\bibinfo{person}{Emir Kamenica} {and}
  \bibinfo{person}{Matthew Gentzkow}.} \bibinfo{year}{2011}\natexlab{}.
\newblock \showarticletitle{Bayesian persuasion}.
\newblock \bibinfo{journal}{\emph{American Economic Review}}
  \bibinfo{volume}{101}, \bibinfo{number}{6} (\bibinfo{year}{2011}),
  \bibinfo{pages}{2590--2615}.
\newblock


\bibitem[\protect\citeauthoryear{Mehta, Dawande, Janakiraman, and
  Mookerjee}{Mehta et~al\mbox{.}}{2019}]%
        {mehta2019sell}
\bibfield{author}{\bibinfo{person}{Sameer Mehta}, \bibinfo{person}{Milind
  Dawande}, \bibinfo{person}{Ganesh Janakiraman}, {and} \bibinfo{person}{Vijay
  Mookerjee}.} \bibinfo{year}{2019}\natexlab{}.
\newblock \showarticletitle{How to sell a dataset? Pricing policies for data
  monetization}.
\newblock \bibinfo{journal}{\emph{Pricing Policies for Data Monetization
  (August 1, 2019)}} (\bibinfo{year}{2019}).
\newblock


\bibitem[\protect\citeauthoryear{Smolin}{Smolin}{2023}]%
        {smolin2023disclosure}
\bibfield{author}{\bibinfo{person}{Alex Smolin}.}
  \bibinfo{year}{2023}\natexlab{}.
\newblock \showarticletitle{Disclosure and pricing of attributes}.
\newblock \bibinfo{journal}{\emph{The Rand Journal of Economics}}
  \bibinfo{volume}{54}, \bibinfo{number}{4} (\bibinfo{year}{2023}),
  \bibinfo{pages}{570--597}.
\newblock


\bibitem[\protect\citeauthoryear{Smolin and Yamashita}{Smolin and
  Yamashita}{2022}]%
        {smolin2022information}
\bibfield{author}{\bibinfo{person}{Alex Smolin} {and} \bibinfo{person}{Takuro
  Yamashita}.} \bibinfo{year}{2022}\natexlab{}.
\newblock \showarticletitle{Information design in concave games}.
\newblock  (\bibinfo{year}{2022}).
\newblock


\bibitem[\protect\citeauthoryear{Yang}{Yang}{2022}]%
        {yangselling}
\bibfield{author}{\bibinfo{person}{Kai~Hao Yang}.}
  \bibinfo{year}{2022}\natexlab{}.
\newblock \showarticletitle{Selling Consumer Data for Profit: Optimal
  Market-Segmentation Design and its Consequences}.
\newblock \bibinfo{journal}{\emph{American Economic Review}}
  (\bibinfo{year}{2022}).
\newblock


\end{thebibliography}

\newpage
\appendix

\section*{Online Appendix}
\setcounter{section}{0}
\section{Characterizing Information Structure of Three-Consumer Cases}\label{threeconsumer-opt-characterization}

In the following, we completely characterize the optimal structure for the case with 3 consumers where the distribution could be arbitrary.

\begin{proposition} \label{optimaldesignfor4consumer}
Consider $3$ consumers with  $x_1 \in (0, 1-\frac{H-L}{2t})$ and $x_2, x_3 \in (1-\frac{H-L}{2t}, 1)$. $\lambda(x_1)\ge 0, \lambda(x_2)\ge 0, \lambda(x_3)\ge 0$ are the probabilities of $3$ consumers. The optimal information structure $\pi$ is characterized as follows: 
    
{\rm{\bf{Case:}} $\lambda(x_1)\frac{H-L}{L} <  \frac{1}{2}\lambda(x_3)$.}

If $1 - \frac{\lambda(x_1) (H-L)}{[\lambda(x_1) + \lambda(x_2)]2t} \ge x_2$, then 
\begin{itemize}
    \item If  $\lambda(x_3)(1-x_3)\ge (x_3-x_2)[\lambda(x_1)+\lambda(x_2)]$, then  $\pi(L,L|x_2) = \pi(L,L|x_3)= 1$.
    \item Otherwise,  $\pi(L,L|x_2)=1$ and $\pi(L,L|x_3)=\max \{\frac{1}{2}, 1-\frac{\lambda(x_1)(H-L)}{\lambda(x_3)L} \}$.
\end{itemize}

 If $1 - \frac{\lambda(x_1) (H-L)}{[\lambda(x_1) + \lambda(x_2)]2t} < x_2$, then 

\begin{itemize}
    \item If $x_3 \le \lambda(x_3) + x_2(1-\lambda(x_3))$. 
    \begin{itemize}
        \item If  $x_3<1-\frac{\lambda(x_1)(H-L)}{2t}$,
 then $\pi(L,L|x_2) = \pi(L,L|x_3)= 1$.
 \item Otherwise,  $\pi(L,L|x_2) = \pi(L,L|x_3)=\max\{\frac{1}{2}, 1- \frac{\lambda(x_1)(H-L)}{L(\lambda(x_2)+\lambda(x_3))} \}$.
    \end{itemize}

    \item If 
 $x_3 > \lambda(x_3) + x_2(1-\lambda(x_3))$.
\begin{itemize}
    \item If  $\lambda(x_1)(H-L)<2(\lambda(x_1) + \lambda(x_2))(1-2x_2+x_3)t-2\lambda(x_3)(1-x_3)t$, then $\pi(L,L|x_3) = 1-\lambda(x_1)\frac{H-L}{\lambda(x_3)L}$ and $\pi(L,L|x_2)=1$; 
    \item Otherwise, $\pi(L,L|x_2)=\pi(L,L|x_3)=\max\{\frac{1}{2}, 1- \frac{\lambda(x_1)(H-L)}{L(\lambda(x_2)+\lambda(x_3))} \}$.
\end{itemize}
\end{itemize}

{\rm{\bf{Case:}}} $\lambda(x_1)\frac{H-L}{L} \ge  \frac{1}{2}\lambda(x_3)$.

If  $1-\frac{\lambda(x_1)}{(1-\lambda(x_3))}\frac{H-L}{2t} \ge x_2$, then
\begin{itemize}
    \item  If $(\lambda(x_1) + \lambda(x_2))(2x_2-2x_3)t + 2\lambda(x_3)(1-x_3)t  \ge 0$, then  $\pi(L,L|x_2)=\pi(L,L|x_3)=1$. 
    \item  Otherwise, $\pi(L,L|x_2)=1$ and $\pi(L,L|x_3)=\frac{1}{2}$. 
\end{itemize}

If  $x_2 > 1-\frac{\lambda(x_1)(H-L)}{2t(1-\lambda(x_3))}$, then
\begin{itemize}
    \item If $x_2 \le 1 - \frac{\lambda(x_1)(H-L)}{2t}$, then
        \begin{itemize}
        \item If $x_3 < 1-\frac{\lambda(x_1)(H-L)}{2t}$, then 
         $\pi(L,L|x_2)=\pi(L,L|x_3)=1$.
        \item If ${\lambda(x_3)} + x_2(1-\lambda(x_3)) \ge x_3 \ge 1-\frac{\lambda(x_1)(H-L)}{2t}$, then $\pi(L,L|x_2)=\pi(L,L|x_3)=\max\{\frac{1}{2}, 1- \frac{\lambda(x_1)(H-L)}{L(\lambda(x_2)+\lambda(x_3))} \}$.
        \item If $1+\frac{\lambda(x_1)(H-L)}{2t} - 2(\lambda(x_1)+\lambda(x_2))(1-x_2) \ge x_3 > {\lambda(x_3)} + x_2(1-\lambda(x_3))$, then $\pi(L,L|x_2)=\pi(L,L|x_3)=\max\{\frac{1}{2}, 1- \frac{\lambda(x_1)(H-L)}{L(\lambda(x_2)+\lambda(x_3))} \}$.
        \item Otherwise, then $\pi(L,L|x_3) = \frac{1}{2}$ and $\pi(L,L|x_2)=\max\{\frac{1}{2}, 1-[\frac{\lambda(x_1)(H-L)}{L\lambda(x_2)} - \frac{\lambda(x_3)}{2\lambda(x_2)}]\}$.
    \end{itemize}

 \item If $x_2 > 1 - \frac{\lambda(x_1)(H-L)}{2t}$, then
    
\begin{itemize}
       \item  {\rm{\bf Case}} ${\lambda(x_3)} + x_2(1-\lambda(x_3)) \ge x_3 \ge 1-\frac{\lambda(x_1)(H-L)}{2t}$. 

       \begin{itemize}
           \item If $\lambda(x_1)\frac{H-L}{L}\ge 1-\lambda(x_1)$, then $\pi(L,L|x_2)=\pi(L,L|x_3)=\max\{0, \frac{[L-H + (2-x_1-x_2)t]}{[L-H + (2-2x_2)t]}\}$.   
           \item If $\lambda(x_1)\frac{H-L}{L}\le \frac{1}{2}(\lambda(x_2)+\lambda(x_3))$,
       then  $\pi(L,L|x_2)=\pi(L,L|x_3)=\max\{\frac{1}{2}, 1- \frac{\lambda(x_1)(H-L)}{L(\lambda(x_2)+\lambda(x_3))} \}$. 
       \item If $1-\lambda(x_1)\ge\lambda(x_1)\frac{H-L}{L} > \frac{1}{2}(\lambda(x_2)+\lambda(x_3))$, then the optimal design  $\pi(L,L|x_2)=\pi(L,L|x_3)=\max\{\frac{[L-H + (2-x_1-x_2)t]}{[L-H + (2-2x_2)t]}, 1- \frac{\lambda(x_1)(H-L)}{L(\lambda(x_2)+\lambda(x_3))} \}$.
       \end{itemize}
       
       %{\rm{\bf i)}} If $\lambda(x_1)\frac{H-L}{L}\ge 1-\lambda(x_1)$, then $\pi(L,L|x_2)=\pi(L,L|x_3)=\max\{0, \frac{[L-H + (1-x_1-x_2)t]}{[L-H + (1-2x_2)t]}\}$.   
       
       %{\rm{\bf ii)}} If $\lambda(x_1)\frac{H-L}{L}\le \frac{1}{2}(\lambda(x_2)+\lambda(x_3))$, then  $\pi(L,L|x_2)=\pi(L,L|x_3)=\max\{\frac{1}{2}, 1- \frac{\lambda(x_1)(H-L)}{L(\lambda(x_2)+\lambda(x_3))} \}$. 
       
       %{\rm{\bf iii)}} If $\lambda(x_1)\frac{H-L}{L} > \frac{1}{2}(\lambda(x_2)+\lambda(x_3))$, then  $\pi(L,L|x_2)=\pi(L,L|x_3)=\max\{\frac{[L-H + (1-x_1-x_2)t]}{[L-H + (1-2x_2)t]}, 1- \frac{\lambda(x_1)(H-L)}{L(\lambda(x_2)+\lambda(x_3))} \}$.
       
    \item  {\rm{\bf Case}} $1+\frac{\lambda(x_1)(H-L)}{2t} - 2(\lambda(x_1)+\lambda(x_2))(1-x_2) \ge x_3 > \lambda(x_3) + x_2(1-\lambda(x_3))$.

    \begin{itemize}
        \item      If $\lambda(x_1)\frac{H-L}{L}\ge 1-\lambda(x_1)$, then  $\pi(L,L|x_2)=\pi(L,L|x_3)=\max\{0, \frac{[L-H + (2-x_1-x_2)t]}{[L-H + (2-2x_2)t]} \}$.
    
   \item   If $\lambda(x_1)\frac{H-L}{L}\le \frac{1}{2}(\lambda(x_2)+\lambda(x_3))$, then $\pi(L,L|x_2)=\pi(L,L|x_3) =\max\{\frac{1}{2}, 1- \frac{\lambda(x_1)(H-L)}{L(\lambda(x_2)+\lambda(x_3))} \}$. 
    
   \item   If $1-\lambda(x_1) \ge\lambda(x_1)\frac{H-L}{L} > \frac{1}{2}(\lambda(x_2)+\lambda(x_3))$, then the optimal design $\pi(L,L|x_2)=\pi(L,L|x_3)=\max\{\frac{[L-H + (2-x_1-x_2)t]}{[L-H + (2-2x_2)t]}, 1- \frac{\lambda(x_1)(H-L)}{L(\lambda(x_2)+\lambda(x_3))} \}$.
    \end{itemize}

    %{\rm{\bf i)}} If $\lambda(x_1)\frac{H-L}{L}\ge 1-\lambda(x_1)$, then  $\pi(L,L|x_2)=\pi(L,L|x_3)=\max\{0, \frac{[L-H + (1-x_1-x_2)t]}{[L-H + (1-2x_2)t]} \}$.
    
    %{\rm{\bf ii)}} If $\lambda(x_1)\frac{H-L}{L}\le \frac{1}{2}(\lambda(x_2)+\lambda(x_3))$, then $\pi(L,L|x_2)=\pi(L,L|x_3) =\max\{\frac{1}{2}, 1- \frac{\lambda(x_1)(H-L)}{L(\lambda(x_2)+\lambda(x_3))} \}$. 
    
    %{\rm{\bf iii)}} If $\lambda(x_1)\frac{H-L}{L} > \frac{1}{2}(\lambda(x_2)+\lambda(x_3))$, then $\pi(L,L|x_2)=\pi(L,L|x_3)=\max\{\frac{[L-H + (1-x_1-x_2)t]}{[L-H + (1-2x_2)t]}, 1- \frac{\lambda(x_1)(H-L)}{L(\lambda(x_2)+\lambda(x_3))} \}$.
    
    \item  {\rm{\bf Case}} $x_3 > 1+\frac{\lambda(x_1)(H-L)}{2t} - 2(\lambda(x_1)+\lambda(x_2))(1-x_2)$. 

    \begin{itemize}
        \item If $\lambda(x_1)\frac{H-L}{L} \ge 1-\lambda(x_1)$, then the optimal design is $\pi(L,L|x_2)=\pi(L,L|x_3)=\max\{0, \frac{[L-H + (2-x_1-x_2)t]}{[L-H + (2-2x_2)t]}\}$.

        \item If  $1-\lambda(x_1) \ge \lambda(x_1)\frac{H-L}{L} \ge \frac{1}{2}(\lambda(x_2)+\lambda(x_3))$, then the optimal design is  $\pi(L,L|x_2)=\pi(L,L|x_3)=\max\{\frac{[L-H + (2-x_1-x_2)t]}{[L-H + (2-2x_2)t]}, 1- \frac{\lambda(x_1)(H-L)}{L(\lambda(x_2)+\lambda(x_3))} \}$.

        \item If $\lambda(x_1)\frac{H-L}{L} <\frac{1}{2}(\lambda(x_2)+\lambda(x_3))$, the optimal design is $\pi(L,L|x_3)= \frac{1}{2}$ and $\pi(L,L|x_2) =  1- [\frac{\lambda(x_1)(H-L)}{\lambda(x_2)L} - \frac{\lambda(x_3)}{2\lambda(x_2)}]$.
    \end{itemize}
    
    %{\rm{\bf i)}}  If $\lambda(x_1)\frac{H-L}{L} \ge 1-\lambda(x_1)$, then the optimal design is $\pi(L,L|x_2)=\pi(L,L|x_3)=\max\{0, \frac{[L-H + (1-x_1-x_2)t]}{[L-H + (1-2x_2)t]}\}$.
       
     %  {\rm{\bf ii)}} If  $\lambda(x_1)\frac{H-L}{L} \ge \frac{1}{2}(\lambda(x_2)+\lambda(x_3))$, then the optimal design is  $\pi(L,L|x_2)=\pi(L,L|x_3)=\max\{\frac{[L-H + (1-x_1-x_2)t]}{[L-H + (1-2x_2)t]}, 1- \frac{\lambda(x_1)(H-L)}{L(\lambda(x_2)+\lambda(x_3))} \}$.
       
     %  {\rm{\bf iii)}} If $\lambda(x_1)\frac{H-L}{L} <\frac{1}{2}(\lambda(x_2)+\lambda(x_3))$, the optimal design is $\pi(L,L|x_3)= \frac{1}{2}$ and $\pi(L,L|x_2) =  1- [\lambda(x_1)\frac{H-L}{L} - \frac{1}{2}\lambda(x_3)]\frac{1}{\lambda(x_2)}$.
\end{itemize}
\end{itemize}
\end{proposition}

Before presenting the proof, we show several characterization results in this case. Note that the information structure in Lemma \ref{lemmasparsestructure} continues to hold. Therefore, we only need to characterize $\pi(H,L|x_i)$ and $\pi(L,L|x_i)$ for $i \in \{2, 3\}$. In the following discussion, we let $y_2\triangleq \pi(L,L|x_2)$ and $y_3\triangleq \pi(L,L|x_3)$. Let $U(x_i)$ denote the consumer $x_i$'s utility of being truthful while $U(x_i \to x_j)$ denote the consumer $x_i$'s utility of misreporting to  $x_j$.

\begin{lemma}\label{x2tox1constraintremoved}
    The constraints $x_2\to x_1$-IC, $x_1\to x_3$-IC and $x_3\to x_1$-IC do not necessarily bind in optimal design.
\end{lemma}
\begin{proof}

    To show $x_1\to x_3$-IC does not necessarily bind, we prove that if $x_1\to x_2$-IC holds, it  directly implies that $x_1\to x_3$-IC holds. Hence, we only need to show 
    \[
    y_2[V_1 -x_1 t-L]+(1-y_2)[V_1 -x_1 t-H]-m_c(x_2) \ge y_3[V_1 -x_1 t-L]+(1-y_3)[V_1 -x_1 t-H]-m_c(x_3)
    \]
    which is equivalent to 
    \[
    (y_2-y_3)(H-L)\ge m_c(x_2) -m_c(x_3)
    \]
    By (\ref{sandwichinequality}) between $x_2$ and $x_3$, we have 
    \[
    2(y_2-y_3)(1-x_2)t \le (y_2-y_3)(H-L)
    \]
    Hence, $x_1\to x_3$-IC does not necessarily bind.

    %Note that in the optimal solution, by Lemma \ref{x1tounderlinexic}, $x_1\to x_2$-IC binds.
    Then, we show that $x_2\to x_1$-IC does not necessarily bind. Note that in the optimal solution, one of the $x_1$-IR  and $x_1\to x_2$-IC constraints must be binding. If $x_1$-IR constraint binds, the RHS of $x_2\to x_1$-IC is $V_1-x_2t-H-m_c(x_1) < 0$. Hence, the constraint $x_2\to x_1$-IC plays no role in the optimal solution. Assume $x_1\to x_2$-IC binds.
    Since
    \begin{align*}
    & U(x_2) - U(x_1\to x_2)\\
        &= \Big\{y_2[V_1-x_2t-L]+(1-y_2)[V_2-(1-x_2)t-L]\Big\} -\Big\{y_2[V_1-x_1t-L] + (1-y_2) [V_1-x_1t-H]\Big\} \\
        &= y_2(x_1-x_2)t + (1-y_2)\Big( H-L -2t+x_1t+x_2t \Big)
    \end{align*}
    and 
    \[U(x_2 \to x_1) - U(x_1) = V_1-x_2t-H - (V_1-x_1t-H) = (x_1-x_2)t\]
    We can see that (by assuming $y_2<1$, otherwise, the inequality becomes equality, which proves the lemma)
    \begin{align*}
    U(x_2) - U(x_1\to x_2) & \ge U(x_2 \to x_1) - U(x_1)\\
        y_2(x_1-x_2)t + (1-y_2)\Big( H-L -2t+x_1t+x_2t \Big) &\ge (x_1-x_2)t \\
         (1-y_2)\Big( H-L -2t+x_1t+x_2t \Big) &\ge (1-y_2)(x_1-x_2)t \\
        H-L -2t+x_1t+x_2t &\ge (x_1-x_2)t \\
         2x_2t &\ge 2t-(H-L)
    \end{align*}
    where the last inequality holds by the definition of $x_2$, and hence we have $U(x_2) - U(x_1\to x_2)  \ge U(x_2 \to x_1) - U(x_1)$  hold.
    This implies that when $x_1\to x_2$-IC binds, i.e., $U(x_1) -m_c(x_1) = U(x_1\to x_2) -m_c(x_2)$, the $U(x_2) -m_c(x_2) \ge U(x_2\to x_1)-m_c(x_1)$ automatically holds.

Finally, we prove that $x_3 \to x_1$-IC does not necessarily bind. If $x_1$-IR binds, then $x_3\to x_1$-IC will satisfy. If If $x_1\to x_2$-IC binds, then $x_3\to x_1$-IC is equivalent to
\[
y_3[V_1-x_3t-L]+(1-y_3)[V_2-(1-x_3)t-L]-m_c(x_3) \ge V_1-y_2L -(1-y_2)H -x_3t - m_c(x_2)
\]
which, by calculations, is equivalent to 
\[
m_c(x_2) - m_c(x_3) \ge (1-y_2)(L-H)+ 2(1-y_3)(1-x_3)t
\]
By (\ref{sandwichinequality}) between $x_2$ and $x_3$, we only need to show that 
\[
2(y_2-y_3)(1-x_3)t\ge (1-y_2)(L-H)+ 2(1-y_3)(1-x_3)t
\]
which is shown to hold by simple calculations.
\end{proof}

 Hence, we only need to focus on the relations between $x_2$ and $x_3$. %As in single-parameter mechanism design literature, we show the binding IC constraint between two adjcent types.
\begin{lemma}\label{x2tox2icbinds}
    $x_2\to x_3$-IC constraint binds at the optimal solution.
\end{lemma}
\begin{proof}
    Note that in the optimal solution, either $x_2$-IR or $x_2\to x_3$-IC constraint binds. Suppose $x_2$-IR constraint binds. Then, $m_c(x_2) = y_2[V_1-x_2t-L]+(1-y_2)[V_2-(1-x_2)t-L]$, by which $x_3 \to x_2$-IC constraint becomes 
    \[
    U(x_3)-m_c(x_3) = y_3[V_1-x_3t-L]+(1-y_3)[V_2-(1-x_3)t-L]-m_c(x_3) \ge (1-2y_2)(x_3-x_2)t
    \]
    If in the optimal solution, $y_2 \ge \frac{1}{2}$, then $x_3$-IR constraint binds, which implies that $x_3\to x_2$-IC always holds. The $x_2 \to x_3$-IC constraint becomes 
    \[
    0 \ge (2y_3-1)(x_3-x_2)t
    \]
    In this case, the broker's revenue is maximized by maximizing $m_c(x_3)$ through maximizing $y_3$. Hence, the RHS of $x_2 \to x_3$-IC reaches $0$ at the optimum.

    If $y_2<\frac{1}{2}$, then $x_3 \to x_2$-IC constraint binds. The $x_2 \to x_3$-IC constraint becomes 
    \[
    0 \ge (2y_3-1)(x_3-x_2)t + (1-2y_2)(x_3-x_2)t = (2y_3-2y_2)(x_3-x_2)t
    \]
    Recall $y_3\le y_2$. If there exists a solution $y_3<y_2$, we can increase $y_3$ to $y_2$. The obedience constraint continues to hold. We can further increase payment $m_c(x_3)$ to the same as $m_c(x_2)$ so that the IC and IR constraints hold, which increases the broker's revenue. The lemma is proved.
\end{proof}

\begin{lemma}\label{oneofthemirbindsgeneral}
    In the optimal design, at least one of the $x_2$-IR and $x_3$-IR constraints must bind.
\end{lemma}
\begin{proof}
 Note that $x_3$-IR constraint or $x_3\to x_2$-IC constraint must bind in the optimal solution. If the $x_3$-IR constraint binds, then we get the observation proved. Instead, suppose $x_3\to x_2$-IC constraint binds. Then, by Lemma \ref{x2tox2icbinds},  we have that $y_2=y_3$ and $m_c(x_2)=m_c(x_3)$. Then, we can increase the payments $m_c(x_1)$, $m_c(x_2)$ and $m_c(x_3)$ simultaneously until one of the IR constraints binds. If $x_2$ or $x_3$-IR constraint binds, then the lemma is proved.  If $x_1$-IR constraint binds first, it implies that the utilities of $x_2$ or $x_3$ misreporting to $x_1$ are negative. Hence, we can further increase the payments $m_c(x_2)$ and $m_c(x_3)$ until one of $x_2$ and $x_3$-IR constraint binds
\end{proof}

Since $x_1$ buys products at price $H$ and $x_2 x_3$ buys products at price $L$, to maximize the broker's revenue, we only need to consider the payments collected from consumers.
\[
R = \lambda(x_1)m_c(x_1) + \lambda(x_2)m_c(x_2) + \lambda(x_3)m_c(x_3)
\]
With this and the above structural results, we show that $x_1\to x_2$-IC constraint binds.

\begin{lemma}\label{x1x2icocnstinaprofof}
    $x_1\to x_2$-IC constraint binds in the optimal solution.
\end{lemma}
\begin{proof}
    If $x_1\to x_2$-IC binds at the optimal solution, we have the lemma proved. Suppose that $x_1$-IR binds.
    By Lemma \ref{oneofthemirbindsgeneral}, we divide the discussions into two cases.  
    
    If $x_2$-IR binds, we actually have the total payment collected from  consumers as 
    \begin{align*}
        R= {}&{} \lambda(x_1)(V_1-x_1t-H) + \lambda(x_2)\Big(y_2[V_1-x_2t-L]+(1-y_2)[V_2-(1-x_2)t-L]\Big) \\
        &{} \quad \quad \quad + \lambda(x_3)[V_1 - t -(1-x_2)t-L+2y_3(1-x_2)t]
    \end{align*}
    Indeed, $x_2\to x_3$-IC binding and $x_3$-IR imply that $y_3\le \frac{1}{2}$.
    Hence, to maximize the broker's revenue, one needs to maximize $y_2$ and $y_3$. Note that the right-hand side of $x_1\to x_2$-IC constraint satisfies
    \[
    y_2[V_1-x_1t-L] + (1-y_2) [V_1-x_1t-H]-m_c(x_2) \le 0
    \]
    which by substituting $m_c(x_2)$ is equivalent to 
       \[
        y_2 (x_2-x_1)t + (1-y_2)[(2-x_1-x_2)t + L -H] \le 0
        \]
        and then
        \[
        (2-x_1-x_2)t + L -H \le y_2[2t-2x_2t+L-H] 
        \]
        \[
        \frac{(x_1+x_2)t - (2t+L-H)}{2x_2t- (2t+L-H)} \ge y_2 
        \]
    Hence, to maximize the revenue, we can set $y_2 =\max\{0,  \frac{(x_1+x_2)t - (t+L-H)}{2x_2t- (t+L-H)}\}$. This indeed implies that $x_1\to x_2$-IC  binds in this case.

%    \jj{need deeper thinking: The maximum total payment is obtained by setting $y_2 =1$. Hence, $x_1\to x_2$ must bind, which in fact implies that $x_1$-IR cannot bind in this case.}

    If $x_3$-IR binds, we have the $x_2$-IR as 
    \[
    (2y_3-1)(x_3-x_2)t \ge 0
    \]
    which implies that $y_3\ge \frac{1}{2}$. 
    If the right-hand side of $x_1\to x_2$-IC constraint is greater than the left-hand side of $x_2$-IR, then $x_1\to x_2$-IC must bind. This requires that 
    \begin{align*}
        &y_2[V_1-x_1t-L] + (1-y_2) [V_1-x_1t-H]-m_c(x_2) \\
        &\hspace{4cm}\ge y_2[V_1-x_2t-L]+(1-y_2)[V_2-(1-x_2)t-L]-m_c(x_2) 
    \end{align*}
        which is equivalent to 
        \[
        \frac{(x_1+x_2)t - (2t+L-H)}{2x_2t- (2t+L-H)} \le y_2 
        \]
    Note that $\frac{(x_1+x_2)t - (2t+L-H)}{2x_2t- (2t+L-H)} \le \frac{1}{2}$ due to $x_1 \le 1 + \frac{L-H}{2t}$. Since we already know that $y_2\ge y_3$ in the optimal solution, we have $y_2\ge \frac{1}{2}$.  Hence, the above inequality holds. 
    Thus, the right-hand side of $x_1\to x_2$-IC is greater than the left-hand side of $x_2$-IR. Hence, $x_1\to x_2$-IC must bind.
\end{proof}

%\jj{=============================}

By that $x_1\to x_2$-IC binds and Lemma \ref{x2tox2icbinds}, we have 
\begin{gather}
    m_c(x_2)-m_c(x_3) = (y_2-y_3)(2-2x_2)t \label{bindingmcx2x3x1_1}\\
m_c(x_2)-m_c(x_1) = y_2(H-L) \label{bindingmcx2x3x1_2}
\end{gather}
The total payments collected from consumers are
\begin{equation}\label{eqtotoalpayment3con}
R =  m_c(x_2) -\lambda(x_1)y_2(H-L) - \lambda(x_3)(y_2-y_3)(2-2x_2)t
\end{equation}

\begin{proof}[Proof of Theorem \ref{optimaldesignfor4consumer}]
By Lemma \ref{oneofthemirbindsgeneral}, we first find the necessary conditions for the two cases, from which we then derive the sufficient conditions.

\paragraph{Part I: Necessary Conditions.} ~

{\bf \noindent Case: $x_2$-IR constraint binds.} By (\ref{eqtotoalpayment3con}), the total payments collected are further extended as  
\begin{equation}\label{totalpaymentinx2irbinding}
\begin{aligned}
    R
    & = y_2[V_1-x_2t-L]+(1-y_2)[V_2-(1-x_2)t-L]-\lambda(x_1)y_2(H-L) - \lambda(x_3)(y_2-y_3)(2-2x_2)t \\
    & = V-(1-x_2)t-L + y_2(2-2x_2)t -t - \lambda(x_1)y_2(H-L)-\lambda(x_3)(y_2-y_3)(2-2x_2)t\\
    & = V-(1-x_2)t-L-t + y_2[(2-2x_2)t-\lambda(x_1)(H-L) - \lambda(x_3)(2-2x_2)t] + y_3\lambda(x_3)(2-2x_2)t
\end{aligned}
\end{equation}

\begin{comment}
First note that the $x_1$-IR constraint must be satisfied. By the binding $x_2$-IR, it implies that  
\begin{align*}
   & y_2[V-x_1t-L] + (1-y_2) [V-x_1t-H] - m_c(x_2) \ge 0\\
   \Longleftrightarrow {} & {} y_2[V-x_1t-L] + (1-y_2) [V-x_1t-H] - \Big\{ y_2[V-x_2t-L]+(1-y_2)[V-(1-x_2)t-L] \Big\} \ge 0 \\
   \Longleftrightarrow {} & {} y_2[x_2t-x_1t] + (1-y_2) [L-H + (1-x_1-x_2)t]  \ge 0 \\
   \Longleftrightarrow {} & {}  [L-H + (1-x_1-x_2)t]  \ge y_2 [L-H + (1-2x_2)t]  \\
   \Longleftrightarrow {} & {}  y_2 \ge \frac{[L-H + (1-x_1-x_2)t]}{[L-H + (1-2x_2)t]}
\end{align*}
Note that as long as $x_2>x_1$ and $x_2\neq \frac{1}{2}-\frac{H-L}{2t}$, we have that 
\begin{align*}
    &\frac{[L-H + (1-x_1-x_2)t]}{[L-H + (1-2x_2)t]} \le \frac{1}{2} \\
    \Longleftrightarrow {} & {} L-H + (1-x_1-x_2)t \ge \frac{1}{2}[L-H + (1-2x_2)t] \\
     \Longleftrightarrow {} & {} \frac{1}{2}(L-H) + (1-x_1-x_2)t \ge  \frac{1}{2}(1-2x_2)t \\
     \Longleftrightarrow {} & {} \frac{1}{2} - \frac{H-L}{2t} \ge x_1
\end{align*}
This implies that if $y_2\le \frac{1}{2}$, it has a lower bound $\frac{[L-H + (1-x_1-x_2)t]}{[L-H + (1-2x_2)t]}$.
\jjr{====}
\end{comment}

By $x_1\to x_2$-IC binding, we have that the $x_1$-IR constraint as 
\begin{align*}
    y_2[V_1-x_1t-L] + (1-y_2) [V_1-x_1t-H] - m_c(x_2) \ge 0
\end{align*}
which is equivalent to 
\[
y_2 \ge \frac{[L-H + (2-x_1-x_2)t]}{[L-H + (2-2x_2)t]}
\]
Further note that $\frac{[L-H + (2-x_1-x_2)t]}{[L-H + (2-2x_2)t]} \le \frac{1}{2}$. This implies that if $y_2\le \frac{1}{2}$, it has a lower bound $\frac{[L-H + (2-x_1-x_2)t]}{[L-H + (2-2x_2)t]}$.

By the binding $x_2$-IR, the $x_3$-IR  and $x_3\to x_2$-IC constraints of $x_3$ become
\begin{align*}
    (1-2y_3)(x_3-x_2)t &\ge 0 \tag{$x_3$-IR}\\
    (1-2y_3)(x_3-x_2)t & \ge (1-2y_2)(x_3-x_2)t \tag{$x_3\to x_2$-IC}
\end{align*}
Hence, $y_3 \le \frac{1}{2}$. Furthermore, the obedience constraint for $s_1=H$ requires that 
%\[\mu(x_1)H \ge [\mu(x_1) + \mu(x_2)(1-y2) + \mu(x_3)(1-y_3)]L \Leftrightarrow \mu(x_1)\frac{H-L}{L} \ge (1-\mu(x_1))(1-y_2) + \mu(x_3)(y_2-y_3)\]
\[
\lambda(x_1)H \ge [\lambda(x_1) + \lambda(x_2)(1-y2) + \lambda(x_3)(1-y_3)]L
\]
To satisfy the obedience constraint, the minimum requirement is $\lambda(x_1)\frac{H-L}{L} \ge \frac{1}{2}\lambda(x_3)$. This also implies the following corollary.
\begin{corollary}\label{x2ircannotbind}
    If $\lambda(x_1)\frac{H-L}{L} <  \frac{1}{2}\lambda(x_3)$, the optimal solution cannot have $x_2$-IR binds.
\end{corollary}
Hence, to maximize (\ref{totalpaymentinx2irbinding}), we  have the following result: If $(1-\lambda(x_3))(2-2x_2)t-\lambda(x_1)(H-L)\ge 0$, i.e., $1-\frac{\lambda(x_1)}{(1-\lambda(x_3))}\frac{H-L}{2t} \ge x_2$, then $y_2=1$ and $y_3=\frac{1}{2}$ is the optimal solution.

 If $(1-\lambda(x_3))(2-2x_2)t-\lambda(x_1)(H-L)< 0$, i.e., $x_2 > 1 - \frac{\lambda(x_1)}{1-\lambda(x_3)} \frac{H-L}{2t}$,  we need to carefully choose $y_2$ and $y_3$, while maintaining the obedience constraint and one of the $x_3$-IR  and $x_3\to x_2$-IC constraints for $x_3$ binding.
\begin{itemize}
    \item $|(1-\lambda(x_3))(2-2x_2)t-\lambda(x_1)(H-L)| \le \lambda(x_3) (2-2x_2)t$, which is equivalent to $x_2 \le 1-\frac{\lambda(x_1)(H-L)}{2t}$. Then, by the obedience constraint, the optimal solution would be $y_3=\frac{1}{2}$ and $y_2 = \max \{\frac{1}{2}, 1- [\frac{\lambda(x_1)}{\lambda(x_2)}\frac{H-L}{L} - \frac{\lambda(x_3)}{2\lambda(x_2)}]\}$.

    %\jj{=========}
    
    \item $|(1-\lambda(x_3))(1-2x_2)t-\lambda(x_1)(H-L)| >\lambda(x_3) (1-2x_2)t$, which is equivalent to $x_2 > 1-\frac{\lambda(x_1)(H-L)}{2t}$. The optimal solution is obtained by making $y_2$  as small as possible and $y_3$ as large as possible.
    \begin{itemize}
        \item If $\lambda(x_1)\frac{H-L}{L} \ge \lambda(x_2) + \lambda(x_3)$, which implies that the obedience constraints are always satisfied, then $y_2=y_3=\max\{0, \frac{[L-H + (2-x_1-x_2)t]}{[L-H + (2-2x_2)t]}\}$.
        
        %then we have $y_2=y_3=0$ \jj{$y_2=y_3=\max\{0, \frac{[L-H + (1-x_1-x_2)t]}{[L-H + (1-2x_2)t]}\}$} as the optimal solution.
        \item If  $ \lambda(x_2) + \lambda(x_3) \ge \lambda(x_1)\frac{H-L}{L} \ge \frac{1}{2}(\lambda(x_2)+\lambda(x_3))$, it implies that setting $y_2 \le \frac{1}{2}$ will not violate the obedience constraint. Since $y_2\ge y_3$,  
     we have the optimal design as $y_2=y_3=\max\{1-\frac{\lambda(x_1)(H-L)}{L(\lambda(x_2)+\lambda(x_3))}, \frac{[L-H + (2-x_1-x_2)t]}{[L-H + (2-2x_2)t]}\}$.
     
    % $y_2=y_3=1-\mu(x_1)\frac{H-L}{L}/(\mu(x_2)+\mu(x_3))$ \jj{$y_2=y_3=\max\{1-\mu(x_1)\frac{H-L}{L}/(\mu(x_2)+\mu(x_3)), \frac{[L-H + (1-x_1-x_2)t]}{[L-H + (1-2x_2)t]}\}$}.
        \item If $\frac{\lambda(x_3)}{2}\le \lambda(x_1)\frac{H-L}{L} <\frac{1}{2}(\lambda(x_2)+\lambda(x_3))$, the optimal solution is $y_3 = \frac{1}{2}$ and $y_2 =  1- [\frac{\lambda(x_1)(H-L)}{L\lambda(x_2)} - \frac{\lambda(x_3)}{2\lambda(x_2)}]$.
    \end{itemize}
\end{itemize}

{\bf Case: $x_3$-IR constraint binds.} 
By the binding $x_3$-IR, (\ref{bindingmcx2x3x1_1}) and (\ref{bindingmcx2x3x1_2}), we have that the total payment collected is 
\begin{equation}\label{x3irbinding}
\begin{aligned}
    R & = V_2-(1-x_3)t-L + y_2 \Big\{ \lambda(x_1)[(2-2x_2)t -(H-L)] + \lambda(x_2)(2-2x_2)t\Big\}  \\
    & \quad \quad \quad + y_3 \Big\{ 2(\lambda(x_1) + \lambda(x_2))(x_2-x_3)t + \lambda(x_3)2(1-x_3)t \Big\} 
\end{aligned}
\end{equation}
Furthermore, we have the $x_2$-IR constraint  equivalent to
\[
(2y_3-1)(x_3-x_2)t \ge 0 
\]
which implies that $y_3\ge \frac{1}{2}$. Note that by the proof of Lemma \ref{x1x2icocnstinaprofof}, we know that in this case, the right-hand side of $x_1\to x_2$-IC is greater than the left-hand side of $x_2$-IR. Hence, the $x_1$-IR automatically satisfies.

If $\lambda(x_1)[(2-2x_2)t -(H-L)] + \lambda(x_2)(2-2x_2)t \ge 0$, i.e., $1 - \frac{\lambda(x_1) (H-L)}{[1-\lambda(x_3)]2t} \ge x_2$, then 
\begin{itemize}
    \item If $(\lambda(x_1) + \lambda(x_2))(2x_2-2x_3)t + 2\lambda(x_3)(1-x_3)t  \ge 0$, i.e., $\lambda(x_3)(1-x_3)\ge (x_3-x_2)[\lambda(x_1)+\lambda(x_2)]$, the optimal solution is to have $y_3 = y_2=1$.
    \item If $\lambda(x_3)(1-x_3)< (x_3-x_2)[\lambda(x_1)+\lambda(x_2)]$, then we should let $y_3$ as small as possible, which is implied by the obedience constraint.  Hence $y_2=1$ and $y_3=\max \{\frac{1}{2}, 1-\frac{\lambda(x_1)(H-L)}{\lambda(x_3)L} \}$.
\end{itemize}

\noindent If $\lambda(x_1)[(2-2x_2)t -(H-L)] + \lambda(x_2)(2-2x_2)t < 0$, i.e., $1 - \frac{\lambda(x_1) (H-L)}{[1-\lambda(x_3)]2t} < x_2$, then 

\begin{itemize}
    \item If $\lambda(x_3)(1-x_3)\ge (x_3-x_2)[\lambda(x_1)+\lambda(x_2)]$, i.e., $x_3 \le \lambda(x_3) + x_2(1-\lambda(x_3))$. 
    \begin{itemize}
        \item If $|\lambda(x_1)[2(1-x_2)t -(H-L)] + 2\lambda(x_2)(1-x_2)t| < (\lambda(x_1) + \lambda(x_2))(2x_2-2x_3)t + 2\lambda(x_3)(1-x_3)t$, i.e., $x_3<1-\frac{\lambda(x_1)(H-L)}{2t}$,  
 the optimal design is $y_2=y_3=1$.

 \item If $|\lambda(x_1)[2(1-x_2)t -(H-L)] + 2\lambda(x_2)(1-x_2)t| \ge (\lambda(x_1) + \lambda(x_2))(2x_2-2x_3)t + 2\lambda(x_3)(1-x_3)t$, i.e., $x_3\ge1-\frac{\lambda(x_1)(H-L)}{2t}$,  then we should let $y_2$ as small as possible. Hence, 
 the optimal design is $y_2=y_3=\max\{\frac{1}{2}, 1- \frac{\lambda(x_1)(H-L)}{L(\lambda(x_2)+\lambda(x_3))} \}$.
    \end{itemize}

    \item Suppose 
$\lambda(x_3)(1-x_3)< (x_3-x_2)[\lambda(x_1)+\lambda(x_2)]$, i.e., $x_3 > \lambda(x_3) + x_2(1-\lambda(x_3))$.
\begin{itemize}
    \item If $|\lambda(x_1)[2(1-x_2)t -(H-L)] + 2\lambda(x_2)(1-x_2)t| < |(\lambda(x_1) + \lambda(x_2))(2x_2-2x_3)t + 2\lambda(x_3)(1-x_3)t|$, 
    {i.e., $\lambda(x_1)(H-L)<2(\lambda(x_1) + \lambda(x_2))(1-2x_2+x_3)t-2\lambda(x_3)(1-x_3)t$,  we have $x_3>1+\frac{\lambda(x_1)(H-L)}{2t} - 2(\lambda(x_1)+\lambda(x_2))(1-x_2)$}. Next, we need to take the obedience constraint into consideration. 
    \begin{itemize}
        \item If $\lambda(x_1)\frac{H-L}{L} < \frac{1}{2}\lambda(x_3)$, then $y_3 = 1-\lambda(x_1)\frac{H-L}{\lambda(x_3)L}$ and $y_2=1$; 
        \item If $\lambda(x_1)\frac{H-L}{L} \ge \frac{1}{2}\lambda(x_3)$, $y_3 = \frac{1}{2}$ and $y_2=\max\{\frac{1}{2}, 1-[\frac{\lambda(x_1)(H-L)}{L\lambda(x_2)} - \frac{\lambda(x_3)}{2\lambda(x_2)}]\}$.
    \end{itemize}

    \item If {$\lambda(x_1)(H-L)\ge 2(\lambda(x_1) + \lambda(x_2))(1-2x_2+x_3)t-2\lambda(x_3)(1-x_3)t$}, then we should let $y_2$ as small as possible. Hence, the optimal design is $y_2=y_3=\max\{\frac{1}{2}, 1- \frac{\lambda(x_1)(H-L)}{L(\lambda(x_2)+\lambda(x_3))} \}$.
\end{itemize}
\end{itemize}
%\noindent 1) Suppose $\frac{1}{2}\mu(x_3)(1-2x_3)\ge (x_3-x_2)[\mu(x_1)+\mu(x_2)]$, i.e., $x_3 \le \frac{\mu(x_3)}{2} + x_2(1-\mu(x_3))$

%We divide it into two cases: 
%i) if $|\mu(x_1)[(1-2x_2)t -(H-L)] + \mu(x_2)(1-2x_2)t| < (\mu(x_1) + \mu(x_2))(2x_2-2x_3)t + \mu(x_3)(1-2x_3)t$, i.e., $x_3<\frac{1}{2}-\frac{\mu(x_1)(H-L)}{2t}$,
% the optimal solution is $y_2=y_3=1$. ii) if $|\mu(x_1)[(1-2x_2)t -(H-L)] + \mu(x_2)(1-2x_2)t| \ge (\mu(x_1) + \mu(x_2))(2x_2-2x_3)t + \mu(x_3)(1-2x_3)t$, i.e., $x_3\ge\frac{1}{2}-\frac{\mu(x_1)(H-L)}{2t}$,  the optimal solution is $y_2=y_3=\max\{\frac{1}{2}, 1- \frac{\mu(x_1)(H-L)}{L(\mu(x_2)+\mu(x_3))} \}$.

%\medskip

%\noindent 2) Suppose 
%$\frac{1}{2}\mu(x_3)(1-2x_3)< (x_3-x_2)[\mu(x_1)+\mu(x_2)]$, i.e., $x_3 > \frac{\mu(x_3)}{2} + x_2(1-\mu(x_3))$

%We divide it into two cases: i) if $|\mu(x_1)[(1-2x_2)t -(H-L)] + \mu(x_2)(1-2x_2)t| < |(\mu(x_1) + \mu(x_2))(2x_2-2x_3)t + \mu(x_3)(1-2x_3)t|$, \jjr{i.e., $\mu(x_1)(H-L)<2(\mu(x_1) + \mu(x_2))(1-2x_2)t-(1-2x_3)t$, $x_3>\frac{1}{2}+\frac{\mu(x_1)(H-L)}{2t} - (\mu(x_1)+\mu(x_2))(1-2x_2)$}, then if $\mu(x_1)\frac{H-L}{L} < \frac{1}{2}\mu(x_3)$, then $y_3 = 1-\mu(x_1)\frac{H-L}{\mu(x_3)L}$ and $y_2=1$; otherwise, $y_3 = \frac{1}{2}$ and $y_2=\max\{\frac{1}{2}, 1-[\frac{\mu(x_1)(H-L)}{L\mu(x_2)} - \frac{\mu(x_3)}{2\mu(x_2)}]\}$  ii) if \jjr{i.e., $\mu(x_1)(H-L)\ge 2(\mu(x_1) + \mu(x_2))(1-2x_2)t-(1-2x_3)t$}, then $y_2=y_3=\max\{\frac{1}{2}, 1- \frac{\mu(x_1)(H-L)}{L(\mu(x_2)+\mu(x_3))} \}$

%\jj{=======================}

\paragraph{Part II: Sufficient Conditions.}~

{\bf Case $\lambda(x_1)\frac{H-L}{L} <  \frac{1}{2}\lambda(x_3)$.} By Corollary \ref{x2ircannotbind}, it only happens that $x_3$-IR binds. Hence, the above necessary conditions become sufficient conditions.

{\bf Case $\lambda(x_1)\frac{H-L}{L} \ge  \frac{1}{2}\lambda(x_3)$.} In this case, either $x_2$-IR or $x_3$-IR can be binding. Notice that the coefficient terms for (\ref{totalpaymentinx2irbinding}) and (\ref{x3irbinding}) of $y_2$ are the same.

If $(1-\lambda(x_3))(2-2x_2)t-\lambda(x_1)(H-L)\ge 0$, i.e., $1-\frac{\lambda(x_1)}{(1-\lambda(x_3))}\frac{H-L}{2t} \ge x_2$, 
\begin{itemize}
    \item  If 
    $(\lambda(x_1) + \lambda(x_2))(2x_2-2x_3)t + 2\lambda(x_3)(1-x_3)t  \ge 0$,
    by the above discussion, we know the optimal solution for $x_2$-IR binding case is $y_3 = \frac{1}{2}$ and $y_2 = 1$, which implies that  the $x_3$-IR binds. Hence, by the binding $x_3$-IR case, the optimal design  is $y_3 = y_2=1$. ($x_3$-IR binds); 
    \item  If $(\lambda(x_1) + \lambda(x_2))(2x_2-2x_3)t + 2\lambda(x_3)(1-x_3)t  < 0$, in either $x_2$-binding or $x_3$-binding case, since $\frac{\lambda(x_1)(H-L)}{\lambda(x_3)L} \ge \frac{1}{2}$, the optimal design is $y_2=1$ and $y_3=\frac{1}{2}$.  
\end{itemize}

If $(1-\lambda(x_3))(1-2x_2)t-\lambda(x_1)(H-L)< 0$, i.e., $1-\frac{\lambda(x_1)}{(1-\lambda(x_3))}\frac{H-L}{2t} < x_2$,  we then compare two thresholds in $x_2$-IR binding and $x_3$-IR binding cases: $x_2 > 1-\frac{H-L}{2t}\frac{\lambda(x_1)}{1-\lambda(x_3)}$ implies  $\lambda(x_3)+x_2(1-\lambda(x_3)) > 1-\frac{\lambda(x_1)(H-L)}{2t}$. 

Furthermore, we have the following relations betwen the thresholds in the $x_3$-IR binding case:  $1+\frac{\lambda(x_1)(H-L)}{2t} - 2(\lambda(x_1)+\lambda(x_2))(1-x_2) > \lambda(x_3) + x_2(1-\lambda(x_3))$.

%Also, we have $\frac{1}{2}+\frac{\mu(x_1)(H-L)}{2t} - (\mu(x_1)+\mu(x_2))(1-2x_2) > \frac{\mu(x_3)}{2} + x_2(1-\mu(x_3)) \Leftrightarrow x_2 > \frac{1}{2}-\frac{\mu(x_1)(H-L)}{2t(1-\mu(x_3))}$.

\begin{itemize}
    \item If $x_2 \le 1 - \frac{\lambda(x_1)(H-L)}{2t}$, then 
    \begin{itemize}
        \item Consider the case $x_3 < 1-\frac{\lambda(x_1)(H-L)}{2t}$.
        Since $y_3 = \frac{1}{2}$ in the $x_2$-IR binding case, we know that $x_3$-IR must binds in this solution. Furthermore, by the $x_3$-IR binding case, we know that by increasing $y_3$, we can further increase the total collected payment. Hence, the optimal solution is $y_2=y_3=1$.
        \item Suppose $\lambda(x_3) + x_2(1-\lambda(x_3)) \ge x_3 \ge 1-\frac{\lambda(x_1)(H-L)}{2t}$. By the same discussion as above, we know that $x_3$-IR must bind. Then, the optimal solution is $y_2=y_3=\max\{\frac{1}{2}, 1- \frac{\lambda(x_1)(H-L)}{L(\lambda(x_2)+\lambda(x_3))} \}$.
        \item If $1+\frac{\lambda(x_1)(H-L)}{2t} - 2(\lambda(x_1)+\lambda(x_2))(1-x_2) \ge x_3 > \lambda(x_3) + x_2(1-\lambda(x_3))$, then $y_2=y_3=\max\{\frac{1}{2}, 1- \frac{\lambda(x_1)(H-L)}{L(\lambda(x_2)+\lambda(x_3))} \}$ by the same discussion.
        \item Otherwise, we have $y_3 = \frac{1}{2}$ and $y_2=\max\{\frac{1}{2}, 1-[\frac{\lambda(x_1)(H-L)}{L\lambda(x_2)} - \frac{\lambda(x_3)}{2\lambda(x_2)}]\}$.
    \end{itemize}

%\jj{===================}
    
       \item If $x_2 > 1 - \frac{\lambda(x_1)(H-L)}{2t}$, we only need to consider $x_3\ge x_2\ge 1-\frac{\lambda(x_1)(H-L)}{2t}$;
       \begin{itemize}
       \item  Suppose $\lambda(x_3) + x_2(1-\lambda(x_3)) \ge x_3 \ge 1-\frac{\lambda(x_1)(H-L)}{2t}$. We compare the optimal solution of $x_3$-IR binding case with the $x_2$-IR binding case. 
       \begin{itemize}
           \item If $\lambda(x_1)\frac{H-L}{L}\ge 1-\lambda(x_1)$, then the optimal solution in $x_3$-IR binding case is $y_2=y_3=\max\{\frac{1}{2}, 1- \frac{\lambda(x_1)(H-L)}{L(\lambda(x_2)+\lambda(x_3))} \} = \frac{1}{2}$. It implies that $x_2$-IR binds. Hence, compared with $x_2$-IR binding case, the optimal design is {$y_2=y_3=\max\{0, \frac{[L-H + (2-x_1-x_2)t]}{[L-H + (2-2x_2)t]}\}$}.
           \item If $\lambda(x_1)\frac{H-L}{L}\le \frac{1}{2}(\lambda(x_2)+\lambda(x_3))$,
       then the optimal solution in $x_2$-IR binding case implies that $x_3$-IR binds (i.e., $y_3=\frac{1}{2}$). Hence, the optimal design is $y_2=y_3=\max\{\frac{1}{2}, 1- \frac{\lambda(x_1)(H-L)}{L(\lambda(x_2)+\lambda(x_3))} \}$. 
       \item If $1-\lambda(x_1) \ge \lambda(x_1)\frac{H-L}{L} > \frac{1}{2}(\lambda(x_2)+\lambda(x_3))$, it means that $y_2=y_3=\frac{1}{2}$ is optimal in the $x_3$-IR binding case, where $x_2$-IR binds. Hence, the optimal design is {$y_2=y_3=\max\{\frac{[L-H + (2-x_1-x_2)t]}{[L-H + (2-2x_2)t]}, 1- \frac{\lambda(x_1)(H-L)}{L(\lambda(x_2)+\lambda(x_3))} \}$}.
       \end{itemize}

    \item  Suppose $1+\frac{\lambda(x_1)(H-L)}{2t} - 2(\lambda(x_1)+\lambda(x_2))(1-x_2) \ge x_3 > \lambda(x_3) + x_2(1-\lambda(x_3))$. 
    \begin{itemize}
        \item If $\lambda(x_1)\frac{H-L}{L}\ge 1-\lambda(x_1)$, then the optimal solution in $x_3$-IR binding case is $y_2=y_3=\frac{1}{2}$. It implies that $x_2$-IR binds. Hence, the optimal solution is {$y_2=y_3=\max\{0, \frac{[L-H + (2-x_1-x_2)t]}{[L-H + (2-2x_2)t]} \}$}. 
        \item    If $\lambda(x_1)\frac{H-L}{L}\le \frac{1}{2}(\lambda(x_2)+\lambda(x_3))$, then the optimal solution in $x_2$-IR binding case implies that $x_3$-IR binds (i.e., $y_3=\frac{1}{2}$).
 Hence, the optimal design is $y_2=y_3=\max\{\frac{1}{2}, 1- \frac{\lambda(x_1)(H-L)}{L(\lambda(x_2)+\lambda(x_3))} \}$. 

 \item If $\lambda(x_1)\frac{H-L}{L} > \frac{1}{2}(\lambda(x_2)+\lambda(x_3))$, it means that $y_2=y_3=\frac{1}{2}$ is optimal in the $x_3$-IR binding case, where $x_2$-IR binds. Hence, the optimal design is {$y_2=y_3=\max\{\frac{[L-H + (2-x_1-x_2)t]}{[L-H + (2-2x_2)t]}, 1- \frac{\lambda(x_1)(H-L)}{L(\lambda(x_2)+\lambda(x_3))} \}$}.
    \end{itemize}

    %then $x_2$-IR binds and the optimal solution is $y_2=y_3= 1- \frac{\mu(x_1)(H-L)}{L(\mu(x_2)+\mu(x_3))}$ \jj{{$y_2=y_3=\max\{\frac{[L-H + (1-x_1-x_2)t]}{[L-H + (1-2x_2)t]}, 1- \frac{\mu(x_1)(H-L)}{L(\mu(x_2)+\mu(x_3))} \}$}}.
    \item  Otherwise, $x_3 > 1+\frac{\lambda(x_1)(H-L)}{2t} - 2(\lambda(x_1)+\lambda(x_2))(1-x_2)$.
    
    Note that if $\lambda(x_1)\frac{H-L}{L} \ge \frac{1}{2}(\lambda(x_2)+\lambda(x_3))$, then $\frac{1}{2}\ge  1-[\frac{\lambda(x_1)(H-L)}{L\lambda(x_2)} - \frac{\lambda(x_3)}{2\lambda(x_2)}]$. It implies in the $x_3$-IR binding case, the optimal solution is $y_2=y_3=\frac{1}{2}$, where the $x_2$-IR also binds. Therefore, we have the following results for this case:
    \begin{itemize}
        \item If $\lambda(x_1)\frac{H-L}{L} \ge 1-\lambda(x_1)$,  the optimal design is {$y_2=y_3=\max\{0, \frac{[L-H + (2-x_1-x_2)t]}{[L-H + (2-2x_2)t]}\}$};
        \item If  $\lambda(x_1)\frac{H-L}{L} \ge \frac{1}{2}(\lambda(x_2)+\lambda(x_3))$, then we have the optimal design as {$y_2=y_3=\max\{\frac{[L-H + (2-x_1-x_2)t]}{[L-H + (2-2x_2)t]}, 1- \frac{\lambda(x_1)(H-L)}{L(\lambda(x_2)+\lambda(x_3))} \}$}. 
        \item If $\lambda(x_1)\frac{H-L}{L} <\frac{1}{2}(\lambda(x_2)+\lambda(x_3))$, the optimal solution is $y_3 = \frac{1}{2}$ and $y_2 =  1- [\frac{\lambda(x_1)(H-L)}{L\lambda(x_2)} - \frac{\lambda(x_3)}{2\lambda(x_2)}]$.
    \end{itemize}
\end{itemize}
    
\end{itemize}

\end{proof}

\section{Proof of Proposition \ref{odefneioasymmeteric}} \label{proofofoporopoinondsymeteri}

\begin{proof}[Proof of Proposition \ref{odefneioasymmeteric}]
In this case, we know the total expected revenue is
\begin{equation}
\begin{split}
m_0+m_1+\sum_{x} \lambda(x)m_c(x)
    & \le \lambda(x_1)\Big\{ \sum_{s\neq (H, L)} \pi(s_1, s_2|x_1)[V-tx_1]+ \pi(H,L|x_1)[V-(1-x_1)t] \Big\} \\
    & \hspace{1cm} + \lambda(x_2) \Big\{\sum_{s\neq (L, H)} \pi(s_1, s_2|x_2)[V-t(1-x_2)] + \pi(L,H|x_2)[V-x_2t] \Big\}
\end{split}
\end{equation}
To maximize revenue upper bound, one can set $\pi(H, L|x_1) = \pi(L, H|x_2) = 0$. Furthermore, we can set $\pi(H,H|x_1)=\pi(H,H|x_2)=0$, which will not change the optimal revenue. Now, both consumers will buy products at price $L$ from the closer sellers. It is not hard to see that when the IR constraints of $x_1$ and $x_2$ bind, the broker gains the maximum possbile revenue. Hence, 
we only need to show that there exist such $x$ and $y$ that the IC constraints are satisfied.
Since two IR constraints bind, we have the two IC constraints as
\begin{align*}
        0&\ge t(1-x_2) - x(1-x_1)t-(1-x)tx_1 \tag{$x_1\to x_2$-IC}\label{x1x2irbingdx1x2iccontsrinless0}\\
        0&\ge tx_1-yx_2t-(1-y)(1-x_2)t  \tag{$x_2\to x_1$-IC}
\end{align*}
That is $x\ge \frac{1-x_1-x_2}{1-2x_1}$ and $y\ge \frac{x_1+x_2-1}{2x_2-1}$. 
The proposition is proved.
\end{proof}

%\newpage
%\include{appendix_reference_only}

\end{document}